\documentclass[12pt,a4paper,reqno]{amsart}
\usepackage{amsfonts}
\usepackage{amsthm}
\usepackage{amsmath}
\usepackage{times}
\usepackage{amscd}
\usepackage[latin2]{inputenc}
\usepackage{t1enc}
\usepackage{enumitem}
\usepackage[mathscr]{eucal}
\usepackage{indentfirst}
\usepackage{graphicx}
\usepackage{graphics}
\usepackage{pict2e}
\usepackage{epic}
\usepackage{esint}
\usepackage[margin=2.9cm]{geometry}
\usepackage{epstopdf} 
\usepackage{verbatim}
\usepackage{cancel}
\usepackage{amssymb}
\usepackage{xcolor}
\usepackage{dsfont}
\usepackage{hyperref}
\usepackage{physics}
\usepackage{cancel}

\definecolor{Greenish}{RGB}{34, 139 , 34}
\definecolor{Blueish}{RGB}{39,64, 139}

\usepackage{hyperref}
\hypersetup{
    colorlinks=true,
    linkcolor=Greenish,
    citecolor=Blueish,
    }

\setitemize{itemsep=5pt, topsep=7pt}
\numberwithin{equation}{section}

\newtheorem{lem}{Lemma}[section]
\newtheorem{cor}{Corollary}[section]
\newtheorem{prop}{Proposition}[section]

\newtheorem{ass}{Assumption}[section]

\newtheorem{theorem}{Theorem}

\newcommand{\bR}{\mathbb{R}}

\newcommand{\cG}{\mathcal{G}}
\newcommand{\cE}{\mathcal{E}}
\newcommand{\cF}{\mathcal{F}}

\newcommand{\cI}{\mathcal{I}}
\newcommand{\h}{{\rm h}}
\newcommand{\ho}{{\rm h_{osc}}}
\newcommand{\psio}{\psi_{\rm osc}}

\newcommand{\Ho}{{\rm H_{osc}}}
\newcommand{\rp}{{\rm p}}
\newcommand{\rv}{{\rm v}}
\newcommand{\cg}{g}
\newcommand{\ch}{h}
\newcommand{\cv}{\rv_{\rm crit}}
\newcommand{\rH}{{{\rm H}}_\alpha}
\newcommand{\rh}{{{\rm h}}}

\DeclareMathOperator{\dist}{dist}

\begin{document}

\title[]{Traveling waves and effective mass for the regularized Landau-Pekar equations}

\author{Simone Rademacher}
\address{Department of Mathematics, LMU Munich, Theresienstrasse 39, 80333 Munich}

\date{\today}

\begin{abstract}
We consider the regularized Landau-Pekar equations with positive speed of sound and prove the existence of subsonic traveling waves. We provide a definition of the effective mass for the regularized Landau-Pekar equations based on the energy-velocity expansion of subsonic traveling waves. Moreover we show that this definition of the effective mass agrees with the definition based on an energy-momentum expansion of low energy states. 


\end{abstract}
\maketitle

\section{Introduction and main results}

The polaron is quasi-particle that models an electron moving through an ionic crystal while interacting with its self-induced polarization field. The polarization field can be either described as a quantum field by the Fr\"ohlich model \cite{Froe} (called quantum polaron) or as a classical field by the Landau-Pekar equations \cite{Landau,Pekar,LP} (called classical polaron). The Landau-Pekar equations describe the polaron as a pair $(\psi, \varphi ) \in H^1 (\bR^3) \times L^2_{\sqrt{\epsilon}} ( \mathbb{R}^3)$ where $\psi$ denotes the $L^2$-normalized wave function of the electron and $\varphi$ the classical field which is for a positive function $\varepsilon >0$ an element of 
\begin{align}
L_{\sqrt{\varepsilon}}^2 ( \mathbb{R}^3 ) := \left\lbrace \varphi \;   \vert \; \| \sqrt{\varepsilon} \varphi \|_2 < \infty  \right\rbrace \; . 
\end{align}
The Landau-Pekar equations are given by the coupled system of differential equations 
\begin{align}
\label{def:dyn}
i \partial_t \psi_t = \rh_{\sqrt{\alpha} \varphi_t}\psi_t, \quad i  \varepsilon^{-1} \partial_t \varphi_t =  \varphi_t + \sqrt{\alpha}  \sigma_{\psi_t} 
\end{align}
where $\alpha>0$ denotes the coupling constant, $m>0$ the electron's mass, 
\begin{align}
\label{def:pot}
\sigma_\psi := (2\pi)^{3/2} \frac{v}{ \varepsilon} \widehat{\varrho}_\psi ,   \quad  \h_{ \varphi} =  - \frac{\Delta}{2m} + V_{\varphi}, \quad \text{with} \quad  V_{\varphi} (x) =  2 \Re \int e^{ik \cdot x} v(k) \varphi (k) dk 
\end{align}
and where $\varrho_\psi : =\vert \psi \vert^2$, 
\begin{align}
\label{eq:varepsilon,v}
\varepsilon = 1, \quad \text{and} \quad v(k) = \frac{1}{\vert k \vert} \; . 
\end{align}
The strong coupling limit is linked with a classical field approximation: For $\alpha \rightarrow \infty$, the classical Landau-Pekar equations can be derived from the quantum dynamics generated by the Fr\"ohlich model \cite{FG,FS,G,LMRSS,LRSS,M}.

\subsection*{Effective mass problem for the Landau-Pekar equations} The dynamics of the polaron is closely related to the outstanding problem of its effective mass: Due to the interaction with the self-induced polarization field, the electron slows down. In physics this phenomenon is described by the emergence of a quasi-particle, the polaron, with an increased effective mass $m_{\rm eff}>0$. 

Based on the classical polaron, Landau and Pekar \cite{Landau,Pekar,LP} formulated a famous quantitative prediction for the effective mass in the strong coupling limit. Their heuristic ideas (described in more detail in \cite{FRS_mass}) rely on the existence of traveling waves of the Landau-Pekar equations, i.e. solutions of \eqref{def:dyn} with initial data $(\psi_\rv,  \varphi_\rv ) \in H^1( \bR^3) \times L^2_{\sqrt{\varepsilon}} ( \bR^3)$,  $\| \psi_\rv \|_2 =1$ satisfying 
\begin{align}
\left( \psi_t (x), \; \varphi_t (k)\right)  = \left( e^{i e_\rv t } \psi_\rv  (x -\rv t ) , \;  e^{i \rv \cdot k t }\varphi_\rv (k )  \right) \;  \label{def:tw}
\end{align}
with phase $e_\rv \in \bR$ and velocity $\rv \in \bR^3$. Traveling waves were, however, conjectured to not exists for $\rv\not= 0$ for the Landau-Pekar equations \cite{FRS_mass} due to a vanishing speed of sound. 

Related to that, the corresponding energy functional to \eqref{def:dyn} does not dominate the total momentum. Thus a computation of the energy as function of conserved total momentum yields a constant function and therefore an infinite mass. Contrarily for the quantum Fr\"ohlich model such an energy-momentum expansion allows to approach the quantum polaron's effective mass (see \cite{BS,LS,MMS} resp. \cite{Varadhan,Beetz,Spohn} for recent progress based on different techniques). 

However for the classical polaron, i.e. the Landau-Pekar equations, neither traveling wave solutions nor an energy-momentum expansion serve for a mathematical rigorous definition of the effective mass. To overcome these problems \cite{FRS_mass} provides a definition of the effective mass that is based on a novel energy-velocity expansion and verifies the quantitative prediction by Landau and Pekar for the classical polaron. 

The goal of this paper is to verify Landau and Pekar's heuristic approach for the effective mass, originally formulated for the non-regularized classical polaron, mathematical rigorously for a regularized classical polaron model, namely the regularized Landau-Pekar equations. 

More precisely, we choose instead of \eqref{eq:varepsilon,v} the functions  $\varepsilon,v$  to be sufficiently regular (see Assumptions \ref{ass:reg},\ref{ass:super} below) and show that subsonic traveling wave solutions with non-vanishing velocity $\rv\not= 0$ do exist (Theorem \ref{thm:tw1}) and serve for a definition of the effective mass (Theorem \ref{thm:mass-tw}). Moreover, the resulting formula agrees with a definition of the effective mass through an energy-momentum expansion (Theorem \ref{thm:mass}) and, furthermore, with results obtained for the quantum regularized Fr\"ohlich model \cite{MS}. 


\subsection*{Regularized Landau-Pekar equations} The regularized Landau-Pekar equations describe more generally a particle moving through an excitable medium. We impose the following assumptions on the functions $\varepsilon,v$ and 
\begin{align}
\cg  = ( 2 \pi)^{3/2} \widehat{\left( v / \varepsilon^{-1/2} \right)} \;  .\label{def:cg}
\end{align}

\begin{ass}[Regularity]
\label{ass:reg}
Let $\varepsilon, v$ be radial with $\varepsilon >0$ and such that  $\cg \in H^2 (\bR^3)$ and $\vert \widehat{g} (k) \vert \geq (1+ \vert k \vert)^{-9/2}$ for all $k \in \mathbb{R}^3$. 
\end{ass}

Furthermore we consider underlying media of positive critical velocity $\cv>0$ formulated in the assumption below. The critical velocity is often referred to as speed of sound of the medium.  
\begin{ass} \label{ass:super}
Let $\varepsilon>0$ satisfy $\inf_{k \in \bR^3} \frac{\varepsilon (k)}{\vert k \vert} := \cv$ for a constant $\cv >0$. 
\end{ass}
We remark that Ass. \ref{ass:reg} and Ass. \ref{ass:super} exclude the (non-regularized) Landau-Pekar equations with $\varepsilon,v$ given by \eqref{eq:varepsilon,v} that, in particular, have vanishing speed of sound with the above definition. 

The dynamical equations \eqref{def:dyn} for $\varepsilon,v$ satisfying Ass. \ref{ass:reg},\ref{ass:super} are well defined (see Lemma \ref{lemma:dyn} below). Moreover the energy functional 
\begin{equation}
\label{def:G}
\cG_\alpha ( \psi, \varphi ) = \expval{\h_{\sqrt{\alpha} \varphi}}{ \psi}  + \| \varepsilon^{1/2}\varphi \|_2^2 \quad \text{with} \quad \|\psi\|_2=1 
\end{equation}
where $h_{\varphi}$ is given by \eqref{def:pot} is preserved along the dynamics. For the regularized Landau-Pekar equations we show that there exists subsonic traveling wave solutions with $0 < \vert \rv \vert < \cv$. 

 \begin{theorem}\label{thm:tw1}
Let $\varepsilon, v$ satisfy Ass. \ref{ass:reg} and Ass. \ref{ass:super} and $\vert \rv \vert  < \cv $. Then there exists a traveling wave solution of the form \eqref{def:tw} with $\rv \not= 0$. 
\end{theorem}
 
Theorem \ref{thm:tw1} follows from Proposition \ref{prop:tw} and is proven in Section \ref{sec:proofs-tw}. 
 
We can not treat the case of supersonic traveling waves  $\vert \rv \vert > \cv$. However we conjecture that supersonic traveling waves do not exist. This conjecture is based on the observation that for $\vert \rv \vert > \cv$, the energy functional does not dominate the total momentum, similarly as for the non-regularized model discussed before. 

\subsubsection*{\textbf{Effective mass problem for the regularized Landau-Pekar equations}} We provide two definitions for the effective mass. The first definition (Theorem \ref{thm:mass-tw}) is based on an energy-velocity expansion of traveling wave solutions and inspired by ideas of Landau and Pekar. The second (Theorem \ref{thm:mass}) is based on an energy-momentum expansion for low energy states. Both definitions lead to the same formula for the effective mass and, in particular, verify the physicists' predictions. 
\subsubsection*{Traveling waves approach} We derive an energy-velocity expansion of low-energy traveling wave solutions \eqref{def:tw} with small velocities. To be more precise we consider states $( \psi, \varphi) \in H^1( \mathbb{R}^3) \times L_{\sqrt{\varepsilon}}^2 ( \mathbb{R}^3)$ 
\begin{enumerate}
\item[(i)] with small energy, i.e. satisfying 
\begin{align}
\mathcal{G}_\alpha( \psi , \varphi )  \leq e_\alpha + \kappa \; \text{for sufficiently small} \quad \kappa >0 \quad \text{(independent of }\; \alpha ) 
\end{align}

\item[(ii)$_\rv$] and which are traveling wave solutions of velocity $\rv$, i.e. let $\rv < \cv $ (uniformly in $\alpha$), then $( \psi_\rv, \varphi_\rv)$ solves \eqref{def:tw} with velocity $\rv$ and with phase $e_\rv \geq - e_\alpha + \rv^2/4$. 
\end{enumerate}
The definition of the effective mass through traveling waves is based on their energy-velocity expansion, i.e. for states of the set 
\begin{align}
\label{def:Iv}
\cI_\rv := \lbrace ( \psi, \varphi) \in H^1( \mathbb{R}^3) \times L_{\sqrt{\varepsilon}}^2 ( \mathbb{R}^3) \vert \;  \text{(i), (ii)$_\rv$ are satisfied} \rbrace \; 
\end{align}
we study the energy expansion 
\begin{align}
E^{\rm TW}_\rv := \inf\lbrace \cG_\alpha ( \psi, \varphi ) \vert \; ( \psi, \varphi ) \in \cI_\rv \rbrace \; 
\end{align}
around the ground state energy  $e_\alpha$.

\begin{theorem}
\label{thm:mass-tw} Let $\varepsilon, v$ satisfy Ass. \ref{ass:reg} and Ass. \ref{ass:super}. Assume that for the pair of ground states $(\psi_\alpha, \varphi_\alpha)$ of $\mathcal{G}_\alpha$ given by \eqref{def:G} where $\varphi_\alpha = - \sqrt{\alpha} \sigma_{\psi_\alpha}$, the minimizer $\psi_\alpha$ is unique up to translations and changes of phase. There exists $\alpha_0 >0$ such that for all $\alpha \geq \alpha_0$ and $\alpha \rv \leq 1$, we have 
\begin{align}
\label{eq:energyexp-1}
E_\rv^{\rm TW}  = e_\alpha +  \left(  m +\frac{ 2 (2\pi)^3 \alpha}{3} \| k v \varepsilon^{-3/2} \; \widehat{\varrho}_{\psi_\alpha}  \|_2^2 \right) \frac{\rv^2}{2} +  O( \alpha \rv^3) \; .
\end{align}
\end{theorem}

The energy expansion of Theorem \ref{thm:mass-tw} (i.e. \eqref{eq:energyexp-1}) is proven in Section \ref{sec:mass}. 

The coefficients of the energy expansion are well defined as 
\begin{align}
 \| k v \varepsilon^{-3/2} \; \widehat{\varrho}_{\psi_\alpha}  \|_2^2 = \| k v \varepsilon^{-3/2} \; \widehat{\varrho}_{\widetilde{\psi}_\alpha}  \|_2^2 
\end{align}
for any 
\begin{align}
\psi_\alpha, \widetilde{\psi}_\alpha \in \Theta ( \psi_\alpha ) := \lbrace e^{i \omega} \psi_\alpha^y =e^{i \omega} \psi_\alpha ( \cdot - y ) \vert \; y \in \mathbb{R}^3, \, \omega \in [0, 2 \pi ) \rbrace \; . 
\end{align} 
We define the effective mass as the second order coefficient of the expansion of $E_\rv^{\rm TW}$ around the ground state energy. It follows from the ground state's approximation (see Proposition \ref{prop:gs} below) that $\widehat{\varrho}_{\psi_\alpha} (k) \rightarrow 1$ point-wise in the limit $\alpha \rightarrow \infty$. Therefore in the strong coupling limit the leading order in $\alpha$ of the effective mass is given by 
\begin{align}
\lim_{\alpha \rightarrow \infty} \alpha^{-1} m_{\rm eff}^{\rm TW} := \lim_{\alpha \rightarrow \infty}  \alpha^{-1} \left( \lim_{\rv \rightarrow 0}  \frac{E_\rv^{\rm TW} - e_\alpha}{\rv^2/2} \right)   =\frac{ 2 (2\pi)^3 }{3} \| k v \varepsilon^{-3/2} \;  \|_2^2 \label{def:mass-TW}
\end{align}
and agrees with the findings of for the quantum Fr\"ohlich model \cite{MS}. 

\subsubsection*{Approach through energy-momentum-expansion} For the second approach we are interested in the infimum of $\mathcal{G}_\alpha$ w.r.t. to the set of states $( \psi, \varphi) \in H^1 ( \mathbb{R}^3) \times L^2_{\sqrt{\varepsilon}} ( \mathbb{R}^3 )$ with small energy (i.e. satisfy (i)) and  
\begin{enumerate}
\item[(ii)$_\rp$] with mean momentum $\rp \in \mathbb{R}^3$, i.e. 
\begin{align}
\label{eq:mean}
 \langle \widehat{\psi} \vert p \vert\widehat{\psi} \rangle + \expval{p}{\varphi} = \rp \; . 
\end{align}
\end{enumerate}
 Thus we consider states of the set 
\begin{align}
\label{def:Ip}
\cI_{\rp}:= \lbrace (\psi, \varphi) \in H^1( \bR^3) \times L^2_{\sqrt{\varepsilon}} ( \mathbb{R}^3) \; \vert  \; \text{such that (i),(ii)$_\rp$ hold} \rbrace \; . 
\end{align}
The definition of the effective mass then relies on an expansion of 
\begin{align*}
E_{\rp}:=\inf \lbrace \cG_\alpha( \psi, \varphi) \vert  \;  \left( \psi, \varphi \right) \in \cI_\rp \rbrace 
\end{align*}
stated in the following theorem. 

\begin{theorem}
\label{thm:mass} 
Let $\varepsilon, v$ satisfy Ass. \ref{ass:reg} and Ass. \ref{ass:super}.Assume that for the pair of ground states $(\psi_\alpha, \varphi_\alpha)$ of $\mathcal{G}_\alpha$ given by \eqref{def:G} where $\varphi_\alpha = - \sqrt{\alpha} \sigma_{\psi_\alpha}$, the minimizer $\psi_\alpha$ is unique up to translations and changes of phase. 
\begin{enumerate}
\item [(a)] There exists $\alpha_0 >0 $ such that for all $\alpha \geq \alpha_0$ and $ \alpha^{-1/4} \rp \leq 1 $ 
\begin{align}
E_\rp  =&  e_\alpha  +  \left(  m +\frac{ 2 (2\pi)^3 \alpha}{3} \| k v \varepsilon^{-3/2} \; \widehat{\varrho}_{\psi_\alpha}  \|_2^2\right)^{-1} \frac{\rp^2}{2} + O( \alpha^{-5/4} \rp^3 )  \; .
\end{align}
\item[(b)] There exists $\alpha_0 >0$ such that for all $\alpha \geq \alpha_0$ and $\alpha^{-1/2}\rp \leq 1$, a pair of minimizers $(\psi_\rp,\varphi_\rp) $ of $E_\rp$ is a traveling wave solution $( \psi_{\rv'},  \varphi_{\rv'} )$ to \eqref{eq:tw} with velocity $\rv' =  m_{\rm eff}^{-1} \rp + O( \alpha^{-3/2} \rp^2)$ and 
\begin{align}
E_\rp = \cG_\alpha ( \psi_{\rv'}, \varphi_{\rv'} ) + O ( \alpha^{-2} \rp^3 ) \; . 
\end{align}

\end{enumerate}
 
\end{theorem}
 
Theorem \ref{thm:mass} \textit{(a),(b)} are proven in Section \ref{sec:mass}. 

We define the effective mass as the coefficient of the second order contribution of the energy-momentum expansion and thus in leading order in $\alpha$ given in the strong coupling limit by 
\begin{align}
\lim_{\alpha \rightarrow \infty} \alpha^{-1} m_{\rm eff}   :=&  \lim_{\alpha \rightarrow \infty} \alpha^{-1}  \lim_{\rp \rightarrow 0 } \left( \frac{E_\rp - e_\alpha }{\rp^2/2} \right)^{-1 }  = \frac{ 2 (2\pi)^3 \alpha}{3} \| k v \varepsilon^{-3/2}  \|_2^2 
\end{align}
which agrees with the effective mass $m_{\rm eff}^{\rm TW}$ defined in \eqref{def:mass-TW} and findings from the quantum Fr\"ohlich model \cite{MS}.

In particular Theorem \ref{thm:mass} {\it (b)} shows that any minimizer of $E_\rp$ is given by a traveling wave solution of velocity $\rv' = m_{\rm eff}^{-1} \rp$, thus, by approximate elements of the set $\cI_\rv$ considered in Theorem \ref{thm:mass-tw}.

We remark that traveling wave solutions for non-linear Schr\"odinger equations with non-vanishing speed of sound were studied in various other settings (see for example \cite{Gravejat} for the Gross-Pitaevksi and \cite{Froehlich} for pseudo-relativistic Hartree equation). We note that \cite{Gravejat} considers a variational approach to traveling waves in the spirit of Theorem \ref{thm:mass} { \em (b)}. 

\subsection*{Structure of the paper} In Section \ref{sec:results-new} we collect properties and approximations of the ground state, ground state energy (Section \ref{sec:results}, Proposition \ref{prop:existence-and-unique} resp. Proposition \ref{prop:gs}) and traveling wave solutions (Section \ref{sec:tw}, Proposition \ref{prop:tw}) that will be important to prove our main theorems. In Section \ref{sec:gs} we prove Propositions \ref{prop:existence-and-unique},\ref{prop:gs} on the ground state's properties. For this we first show the existence of ground states for all $\alpha >0$ in Section \ref{subsec:exist}, then the approximation of the ground (state) by the harmonic oscillator in Section \ref{subsec:approx} and finally the positivity of the Hessian for large $\alpha > \alpha_0$ yielding coercivity estimates. We combine those results in Section \ref{sec:proos-props} to finally prove Propositions \ref{prop:existence-and-unique},\ref{prop:gs}. 
In Section \ref{sec:proofs-tw} we prove afterwards Proposition \ref{prop:tw} (yielding in Theorem \ref{thm:tw1}) on the properties of traveling waves. In Section \ref{sec:mass} we finally prove Theorem \ref{thm:mass-tw} and Theorem \ref{thm:mass} on the two definitions of the effective mass bases on the results before.

\section{Properties of the ground state and traveling waves}
\label{sec:results-new}

\subsection{Properties of the ground state} 
\label{sec:results}

For the regularized polaron's ground state 
\begin{align}
\label{def:gs}
e_\alpha := \inf_{\psi,\varphi} \cG_\alpha ( \psi, \varphi ) \;  
\end{align}
the infimum can be taken first w.r.t. to the phonon field yielding by a completion of the square to the choice 
\begin{align}
\label{def:phi-alpha}
\varphi_\alpha :=  - \sqrt{\alpha} \sigma_{\psi}, \quad \text{where} \quad \sigma_\psi := (2\pi)^{3/2} \frac{v}{ \varepsilon} \widehat{\varrho}_\psi
\end{align}
with $\varrho_\psi = \vert \psi \vert^2$. The resulting energy functional for  $\psi \in H^1 ( \mathbb{R}^3)$ is 
\begin{align}
\label{def:E}
\cE_\alpha ( \psi ) := \inf_{\varphi} \mathcal{G}_\alpha ( \psi, \varphi) =  \expval{- \tfrac{\Delta}{2m} - \alpha \left( \ch * \vert \psi \vert^2\right)}{\psi} \quad \mathrm{with} \quad \ch = ( 2 \pi)^{3/2} \widehat{\left( v^2 / \varepsilon^{-1} \right)} \; . 
\end{align}
Thus if $\psi_\alpha$ is an element of the manifold of minimizers $\mathcal{M}_{\mathcal{E}_\alpha}$ of the energy functional $\mathcal{E}_\alpha$ defined by 
\begin{align}
\mathcal{M}_{\mathcal{E}_\alpha} : = \lbrace \psi \in H^1 ( \mathbb{R}^3) \; \vert \| \psi \|_2=1,  \; \mathcal{E}_\alpha ( \psi_\alpha ) = e_\alpha \rbrace \; , 
\end{align}
then the pair $( \psi_\alpha, \varphi_\alpha)$ with $\varphi_\alpha$ given by \eqref{def:phi-alpha} is an element of the manifold of minimizers $ \mathcal{M}_{\mathcal{G}_\alpha}$ of the energy functional $\mathcal{G}_\alpha$ 
\begin{align}
\mathcal{M}_{\mathcal{G}_\alpha} := \lbrace ( \psi, \varphi ) \in H^1 ( \mathbb{R}^3) \times L^2_{\sqrt{\varepsilon}} ( \mathbb{R}^3) \; \vert \; \mathcal{G}_\alpha ( \psi, \varphi) = e_\alpha\rbrace. 
\end{align}
The energy functional $\mathcal{E}_\alpha$ is symmetric with respect to translations and changes of the phase of the wave function. Thus for any minimizer $\psi_\alpha$ of $\cE_\alpha$ (i.e. $\psi_\alpha \in \mathcal{M}_{\mathcal{E}_\alpha}$) it follows $\Theta ( \psi_\alpha) \subseteq \mathcal{M}_{\mathcal{E}_\alpha}$ where 
\begin{align}
\label{def:Theta}
\Theta ( \psi_\alpha ) := \lbrace e^{i \omega} \psi_\alpha^y =e^{i \omega} \psi_\alpha ( \cdot - y ) \vert \; y \in \mathbb{R}^3, \, \omega \in [0, 2 \pi ) \rbrace \; . 
\end{align}

For the (non-regularized) Pekar functional corresponding to \eqref{eq:varepsilon,v}, the existence of a unique pair of ground states $( \psi_{\rm Pekar}, \varphi_{\rm Pekar})$ up to phases and translations was proven \cite{L} for all $\alpha >0$. For the regularized model we prove the existence of a ground state for all $\alpha >0$. 

\begin{prop}[Existence] \label{prop:existence-and-unique}
Let $\varepsilon, v$ satisfy Assumption \ref{ass:reg}. 
 For all $\alpha >0$ there exists a pair of minimizers $ ( \psi_\alpha, \varphi_\alpha) \in \mathcal{M}_{\mathcal{G}_\alpha}$ with $\varphi_\alpha = -\sqrt{\alpha} \sigma_{\psi_\alpha}$ and $0< \psi_\alpha \in C^\infty ( \mathbb{R}^3)$ satisfying the Euler-Lagrange equation 
\begin{align}
\left( \h_{ \sqrt{\alpha}\varphi_\alpha} - \mu_{\psi_\alpha} \right) \psi_\alpha = 0, \quad  \text{with} \quad \mu_{\psi_\alpha} := \expval{{\rm h}_{\sqrt{\alpha}\varphi_\alpha}}{\psi_\alpha} \; . 
\end{align}
\end{prop}

We remark that the ground state's uniqueness for the regularized model is in general not know. For technical reasons, we can not prove uniqueness up to translations and phases for large $\alpha > \alpha_0$. For that, a refined approximation of the ground state than the one in Proposition \ref{prop:gs} is needed to conclude ground state's uniqueness up to translations and phase by the local coercivity estimates in Corollary \ref{cor:loc-coerc} for large $\alpha > \alpha_0$. 

Note that the first part of Theorem \ref{thm:mass-tw} and Theorem \ref{thm:mass} immediately follow from Proposition \ref{prop:existence-and-unique} that is proven in Section \ref{sec:proos-props}. 

The ground state $\psi_\alpha$'s properties for large coupling constants $\alpha >\alpha_0$ results from the asymptotic behavior of the energy functional $\cG_\alpha$. In fact in the strong coupling limit $\alpha \rightarrow \infty$ the ground state energy $e_\alpha = \cG_\alpha (\psi_\alpha)$ is well described through the harmonic oscillator 
 \begin{align}
\label{def:osc}
\ho := - \frac{\Delta}{2m} + \frac{m \omega^2x^2}{2} , \quad \text{with frequency} \quad  \omega^2 = \frac{ \alpha \| \nabla g\|_2^2}{3 m } \; .
\end{align}
Furthermore its well known ground state 
\begin{align}
\label{def:min-osc}
\psio (x) := \left( \frac{m \omega }{ \pi }\right)^{3/4} e^{-m\omega x^2/2} 
\end{align}
approximates the true ground state of $\cG_\alpha$ as the following Lemma shows.

\begin{prop}[Approximation of the ground state]
\label{prop:gs}Let $\varepsilon,v$ satisfy Assumption \ref{ass:reg}. There exists $\alpha_0 >0 $ and constants $C_1,C_2 >0$ (independent of $\alpha$) such that
\begin{align}
\label{eq:approx-gse}
C_1 \alpha^{1/3} \leq e_\alpha + \alpha \| \cg \|_2^2 - \frac{\sqrt{3 \alpha} \| \nabla \cg \|_2^2}{2 \sqrt{m}} \leq C_2  \;  \quad \text{for all} \quad \alpha \geq \alpha_0 \; . 
\end{align}
Furthermore let $\psi_\alpha \in \mathcal{M}_{\mathcal{E}_\alpha}$. There exists $C_3>0$ (independent of $\alpha$) such that for all $\alpha \geq \alpha_0$
\begin{align}
\label{eq:approx-gs}
\dist_{L^2} \left( \Theta ( \psi_\alpha), \psi_{\rm osc} \right) \leq C_3 \alpha^{-1/20},  \; \quad \dist_{H^1} \left( \Theta ( \psi_\alpha), \psi_{\rm osc} \right) \leq C_3 \alpha^{9/40} \; .  
\end{align}
\end{prop}

Proposition \ref{prop:gs} is proven in Section \ref{sec:proos-props}. 
 
Here we introduced the norm 
\begin{align}
\dist_{L^2} \left( \Theta ( \psi_\alpha ) ,  \; \psi_{\rm osc} \right)  := \inf_{y', \theta'} \| e^{i \theta'}\psi_\alpha^{y'} - \psi_{\rm osc} \|_2
\end{align}
(and similarly for the $H^1$-norm) quantifying the distance of an element $\psi_\alpha$ of the manifold of minimizers to the harmonic oscillator's ground state. 

We remark that the rate of convergence of \eqref{eq:approx-gs} depends for technical reasons on Ass. \ref{ass:reg} namely the regularity of the function $g$.  

\subsection{Traveling waves} 
\label{sec:tw}

The dynamical equations corresponding to the energy functional $\cG_\alpha$ in \eqref{def:G} are given for $(\psi_t, \varphi_t ) \in H^1 (\bR^3) \times L^2_{\sqrt{\varepsilon}} ( \bR^3)$ by the system of coupled partial differential equations \eqref{def:dyn}. The dynamical equations are well-posed as the following Lemma shows. 

\begin{lem}
\label{lemma:dyn}
Let $\varepsilon, v$ satisfy Ass. \ref{ass:reg}. For any $\left( \psi_0, \varphi_0 \right) \in H^1 ( \mathbb{R}^3) \times L_{\sqrt{\varepsilon}}^2 ( \mathbb{R}^3)$ there exists a unique global solution of \eqref{def:dyn}. Furthermore, 
\begin{align}
\mathcal{G}_\alpha( \psi_0,\varphi_0 ) = \mathcal{G}_\alpha( \psi_t, \varphi_t ), \quad \text{and} \quad \| \psi_t \|_2 = \| \psi_0 \|_2 
\end{align}
and there exists $C>0$ such that for $(\psi_0, \varphi_0)$ with $\mathcal{G}( \psi_0, \varphi_0) \leq C \alpha$ we have for all $t \in \mathbb{R}$
\begin{align}
\| \nabla \psi_t \|_2 \leq C \sqrt{\alpha} \quad \text{and} \quad \| \varphi_t \|_{L^2_{\sqrt{\varepsilon}}} \leq C \sqrt{\alpha} \; . 
\end{align}
\end{lem}

The proof of the Lemma follows similarly to \cite[Lemma 2.1]{FG} considering the non-regularized Landau-Pekar equations. The arguments presented in \cite{FG} apply for the regularized case, too, so that we refer for the proof of Lemma \ref{lemma:dyn} to \cite[Lemma 2.1]{FG}.

A traveling wave of velocity $\rv \in \bR^3$ is a solution of \eqref{def:dyn} with initial data $(\psi_\rv,  \varphi_\rv ) \in H^1( \bR^3) \times L^2_{\sqrt{\varepsilon}} ( \bR^3)$,  $\| \psi_\rv \|_2 =1$ satisfying \eqref{def:tw}. The existence of subsonic traveling waves is given by the following Proposition.  
 
 \begin{prop}\label{prop:tw}
Let $\varepsilon, v$ satisfy Ass. \ref{ass:reg} and Ass. \ref{ass:super}. Furthermore assume that $\vert \rv \vert  < \cv $.

\begin{enumerate}
\item[(a)] There exists a traveling wave solution of the form \eqref{def:tw}.

\item[(b)] Furthermore assume that  $\mathcal{G}( \psi_\rv, \varphi_\rv) - e_\alpha \leq \kappa $ for sufficiently small $\kappa>0$ (independent of $\alpha$) and $e_\rv \geq - e_\alpha + \rv^2/4$. Assume that the ground state $\psi_\alpha$ of $\mathcal{E}_\alpha$ is unique up to translations and rotations.  Then there exists $\alpha_0 >0$ and a constant $C>0$ (independent of $\alpha, \rv$) such that 
\begin{align}
\label{eq:tw-cont}
\dist_{L^2} \left( \Theta ( \psi_\alpha ) ,   \; \psi_\rv \right) \leq C \vert  \rv\vert  \; 
\end{align}
for all $\alpha \geq \alpha_0$ and  $\vert \rv \vert  \leq 1$. 
\end{enumerate}
\end{prop}  

Note that Theorem \ref{thm:tw1} follows immediately from Proposition \ref{prop:tw} \textit{(a)}. The proof of Proposition \ref{prop:tw} is given in Section \ref{sec:tw}.

Furthermore note that a similar approximation as in \eqref{eq:tw-cont} holds for the field $\varphi_\rv$, too. For this we remark that instead of minimizing w.r.t. to the field $\varphi$ in \eqref{def:gs} first (as explained in Section \ref{sec:results}) we can take the infimum w.r.t. to the wave function $\psi$ first, too yielding the functional 
\begin{align}
\mathcal{F}_\alpha( \varphi ) = \inf_{\psi} \mathcal{G}_\alpha ( \psi, \varphi), \quad \mathcal{M}_{\mathcal{F}_\alpha} = \lbrace \varphi \in L^2_{\sqrt{\varepsilon}} ( \mathbb{R}^3) \; \vert  \mathcal{F}_\alpha ( \varphi_\alpha ) = e_\alpha \rbrace \; . 
\end{align}
We remark that by the energy functional's symmetries for any $\varphi_\alpha \in \mathcal{M}_{\mathcal{F}_\alpha}$ and 
\begin{align}
\Omega (\varphi_\alpha ) := \lbrace  e^{i z\cdot } \varphi_\alpha \vert \; z \in \mathbb{R}^3 \rbrace 
\end{align}
it follows that $\Omega ( \varphi_\alpha ) \subseteq  \mathcal{M}_{\mathcal{F}_\alpha}$. Then under the same assumption as in Proposition \ref{prop:tw} there exists $C>0$ (independent of $\alpha$) such that 
\begin{align}
\label{eq:tw-cont-2}
\dist_{L^2_{\sqrt{\varepsilon}}} \left( \Omega ( \varphi_\alpha ), \; \varphi_\rv \right)  \leq C \sqrt{\alpha} \vert \rv \vert \; . 
\end{align}

We remark that Proposition \ref{prop:tw} {\it (a)} shows that subsonic traveling waves exist for all $\alpha >0$. However the approximations \eqref{eq:tw-cont}, \eqref{eq:tw-cont-2} of the second part of the Theorem holds for sufficiently large $\alpha \geq \alpha_0$ only. The restriction to sufficiently large $\alpha >0$ in part {\it (b)} ensures the validity of the global coercivity estimates (see Corollary \ref{cor:coerc}) that are proven for sufficiently large $\alpha > \alpha_0$ only. Furthermore we notice that for $\rv =0$ the pair of ground states $( \psi_\alpha, \varphi_\alpha)$ provide a traveling wave solution with $e_\rv = e_\alpha$. In particular the assumption on the phase from part {\it (b)}, made for technical reasons only, is satisfied for $\rv =0$.

\section{Properties of the energy functional $\cE_\alpha$} 
\label{sec:gs}

In this section, we prove Proposition \ref{prop:existence-and-unique} and Proposition \ref{prop:gs} on the properties of the energy functional $\mathcal{G}_\alpha$. 

The proof of Proposition \ref{prop:existence-and-unique} relies on a comparison of properties of $ \mathcal{E}_\alpha$ with the properties of $\ho$ in the limit $\alpha \rightarrow \infty$. Then existence and uniqueness (up to translations and changes of the phase) for pairs of minimizers $( \psi_\alpha, \varphi_\alpha)$ of $\cG_\alpha$ follow with the choice $\varphi_\alpha = - \sqrt{\alpha} \sigma_{\psi_\alpha}$. 

First, in Lemma \ref{lemma:existence}, we prove the existence of a minimizer $\psi_\alpha$ for all $\alpha$. Next we show the ground state (energy) is well approximated through the harmonic oscillator (Lemma \ref{lemma:approx-E}). This approximations allows to show that the Hessian modulo its zero modes of $\mathcal{E}_\alpha$ is asymptotically for $\alpha \rightarrow \infty$ characterized by the harmonic oscillator, and thus positive (Lemma \ref{lemma:Hessian}). This fact has several consequences: We infer first local (Cor. \ref{cor:loc-coerc}) and later global coercivity estimates (Cor. \ref{cor:coerc}) for sufficiently large $\alpha \geq \alpha_0$. For the latter we assume that the ground state $\psi_\alpha$ of $\mathcal{M}_{\mathcal{E}_\alpha}$ is unique up to translations and phase. Furthermore we obtain that the ground state energy $e_\alpha$ is separated from the first excited eigenvalue by a gap of order $\sqrt{\alpha}$ (Cor. \ref{cor:gap}). 

We remark that the strategy for the proofs in this section follow \cite[Section 3]{FS} considering the non-regularized Pekar functional on a Torus of length $L$. For sufficiently large $L$ the uniqueness of the ground state and coercivity estimates are proven based on a comparison with the non-regular Pekar functional defined on the full space for which these properties are well known.

\subsection{Existence} \label{subsec:exist} First we show the existence of minimizers of $\cE_\alpha$ for all $\alpha >0$ in the subsequent Lemma. 

\begin{lem} \label{lemma:existence}
Let $\varepsilon, v$ satisfy Ass. \ref{ass:reg}. For all $\alpha >0$, there exists a minimizer $0< \psi_\alpha \in C^\infty ( \mathbb{R}^3)$ of the functional $\mathcal{E}_\alpha$ satisfying the Euler-Lagrange equation
\begin{align}
\left( \h_{ \sqrt{\alpha}\varphi_\alpha} - \mu_{\psi_\alpha} \right) \psi_\alpha = 0, \quad  \text{with} \quad \mu_{\psi_\alpha} := \expval{h_{\sqrt{\alpha}\varphi_\alpha}}{\psi_\alpha} \; . 
\end{align}
\end{lem}

\begin{proof}
Since $h = g* g$, we have  $\| h \|_{\infty} \leq C \| g \|_2^2 $ and 
\begin{align}
\expval{\left( \ch * \vert \psi \vert^2\right)}{\psi} \leq C \| \cg \|_2^2 \| \psi \|_2^4 
\end{align}
so that by Ass. \ref{ass:reg} there exists $C>0$ such that 
\begin{align}
\mathcal{E}_\alpha ( \psi ) \geq  \tfrac{1}{2m}  \| \nabla \psi \|_2^2  - C \alpha  \; . \label{eq:bound-E}
\end{align}
From \eqref{eq:bound-E} we infer on one hand that $e_\alpha \geq -C \alpha $ for all $\alpha$.  On the other hand, in order to prove the existence of a minimizer, we remark that \eqref{eq:bound-E} shows that any minimizing sequence $\left( \psi_n \right)_{n \in \mathbb{N}}$  is bounded in $H^1$, uniformly in $n \in \mathbb{N}$.  For this reason the sequence by \cite[Lemma 6]{Lieb} resp. \cite[Theorem 8.10]{LiebLoss}, there exists a sequence $\left( y_n \right)_{n \in \mathbb{N}} \in \mathbb{R}^3$ such that the translated sequence $\left( \psi_{n}^{y_n} \right)_{n \in \mathbb{N}}$ has a sub-sequence  $\left( \psi_{n_j} \right)_{n_j \in \mathbb{N}} := \left( \psi_{n_j}^{y_{n_j}} \right)_{n_j \in \mathbb{N}}$ that converges weakly in $H^1( \mathbb{R}^3)$ to a non-zero function. It follows from the Sobolev inequality that this sub-sequence converges strongly in $L^p$ for $2 \leq  p \leq 6$  to a non-zero limit. The limiting function $\psi_\alpha \in H^1$ is again $L^2$-normalized and, moreover, satisfies 
\begin{align}
\expval{- \tfrac{\Delta}{2m}}{\psi_\alpha} \leq \lim_{{n_j} \rightarrow \infty} \expval{- \tfrac{\Delta}{2m}}{\psi_{n_j}} 
\end{align}
by semi-lower continuity of the $H^1$-norm.  Since 
\begin{align}
&\left\vert \expval{\left( \ch * \vert \psi_{n_j} \vert^2 \right) }{\psi_{n_j}} - \expval{\left( \ch * \vert \psi_\alpha \vert^2 \right)}{\psi_\alpha} \right\vert \notag\\
&\quad \leq \left\vert  \bra{\psi_{n_j}} \left( \ch * \vert \psi_{n_j} \vert^2 \right) \ket{\psi_\alpha - \psi_{n_j}} \right\vert +  \left\vert  \bra{\psi_\alpha} \left( \ch * \vert \psi_n \vert^2 \right) \ket{\psi_\alpha - \psi_{n_j}} \right\vert \notag\\
&\quad\quad + \left\vert  \bra{\psi_{n_j}} \left( \ch * \vert \psi_\alpha \vert^2 \right) \ket{\psi_\alpha - \psi_{n_j}} \right\vert + \left\vert  \bra{\psi_\alpha} \left( \ch * \vert \psi_\alpha \vert^2 \right) \ket{\psi_\alpha - \psi_{n_j}} \right\vert
\end{align}
and $\| \left( \ch * \vert \psi_1 \vert^2 \right) \psi_2 \|_{2} \leq C \| \psi_1 \|_{2}^2 \| \psi_2 \|_{2} $ for any $\psi_1,\psi_2 \in L^2$, we find 
\begin{align}
&\left\vert \expval{\left( \ch * \vert \psi_{n_j} \vert^2 \right) }{\psi_{n_j}} - \expval{\left( \ch * \vert \psi_\alpha \vert^2 \right)}{\psi_\alpha} \right\vert \leq C \| \psi_\alpha - \psi_{n_j} \|_2 \rightarrow 0 
\end{align}
as $n \rightarrow \infty$.  Therefore,  
\begin{align}
\mathcal{E}_\alpha ( \psi_\alpha ) \leq \liminf_{{n_j} \rightarrow \infty} \mathcal{E}_{\alpha} (\psi_{n_j} ) = e_\alpha 
\end{align}
and with $\mathcal{E}_ \alpha ( \psi_\alpha ) \geq \liminf_{{n_j} \rightarrow \infty} \mathcal{E}_\alpha ( \psi_{n_j}) = e_\alpha$ (as $( \psi_{n_j} )_{{n_j} \in \mathbb{N}}$ is a minimizing sequence), we conclude that $\mathcal{E}_{\alpha} ( \psi_\alpha ) = e_\alpha$ and, thus, $\psi_\alpha$ is a minimizer.  By invariance of $\cE_\alpha$ w.r.t. to translations and phase, any element of $\Theta (\psi_\alpha)$ defined in \eqref{def:Theta} is a minimizer, too. The positivity and regularity properties of $\psi_\alpha$ follows by standard bootstrap arguments (see for example \cite[Lemma 3.3]{FS}). 
\end{proof}

\subsection{Approximation} \label{subsec:approx} Next we prove that in the strong coupling limit the spectrum of $\mathcal{E}_\alpha$ is well approximated by the harmonic oscillator $\ho$'s spectrum. The idea is to use the Taylor expansion of the potential given by Ass. \ref{ass:reg} through 
\begin{align}
h( x) =  \int \frac{v^2(k)}{\varepsilon(k)} e^{i k \cdot x} dk  = \int \frac{v^2(k)}{\varepsilon(k)} \cos ( k \cdot x) dk \;   \label{eq:taylor-h}
\end{align}
to show its asymptotic quadratic behavior. The ground state energy of $\ho$ (as defined in \eqref{def:osc}) is well known 
\begin{align}
e_{\rm osc} = \frac{\sqrt{3\alpha}\| \nabla g \|_2^2}{ 2\sqrt{m}}
\end{align}
and separated from the rest of the spectrum by a gap of order $\sqrt{\alpha}$. For this we compare the Hessian of $\mathcal{E}_\alpha$ with 
\begin{align}
\label{def:hessian-osc}
\mathcal{H}_{\rm osc}:= \inf_{\substack{f \in H^1 ( \mathbb{R}^3), \| f \|_2 =1 \\ f \in {\rm span} \lbrace \psio \rbrace^\perp }} \expval{\Ho}{f} , \quad \text{with} \quad \Ho = \ho - e_{\rm osc}
\end{align}
that is known to be positive, and thus, yielding coercivity estimates of the form 
\begin{align}
\label{eq:coerc-osc}
\expval{\ho}{f} - e_{\rm osc} \geq C \alpha^{1/2} \inf_{\theta \in (0, 2 \pi]}\| e^{i \theta} \psio -  f \|_2^2 \notag \\
\expval{\ho}{f} - e_{\rm osc} \geq C \alpha^{1/4} \inf_{\theta \in (0, 2 \pi]}\| e^{i \theta} \psio -  f \|_{H^1 ( \mathbb{R}^3)}^2 \; . 
\end{align}
Furthermore we compare the deviation of the ground state of $\psi_\alpha$ with the one of the harmonic oscillator that is known to be $\| x^2 \psio \|_2 = C \alpha^{-1/2}$ for some $C>0$.

\begin{lem}
\label{lemma:approx-E}
Let $\varepsilon, v$ satisfy Ass. \ref{ass:reg}. 

\begin{enumerate}
\item[(a)] Then there exists $C_1,C_2,C_3>0$ (independent of $\alpha$) such that 
\begin{align}
\label{eq:spec}
C_1 \alpha^{1/3} \leq e_\alpha +  \alpha \| \cg \|_2^2 - \frac{\sqrt{3 \alpha} \| \nabla \cg \|_2^2}{2 \sqrt{m}} \leq C_2 \; 
\end{align}
and 
\begin{align}
\vert \mu_{\psi_\alpha} - \mu_{\psio} \vert \leq C_3 \alpha^{5/12}  \;  
\end{align} 
where we introduced the notation $\mu_{\psio} = \expval{- \Delta + 2 \alpha (h * \vert \psio \vert^2)}{\psio}$ \; . 
Let $\psi_\alpha \in \mathcal{M}_{\mathcal{E}_\alpha}$ such that 
\begin{align}
\label{eq:ass-gs}
\dist_{L^2} \left( \Theta ( \psi_\alpha), \psi_{\rm osc} \right) = \| \psio - \psi_\alpha \|_2  \; . 
\end{align}
Then there exists  $\alpha_0 >0$ and $C>0$ (independent of $\alpha$) such that for all $\alpha \geq \alpha_0$ we have 
\begin{align}
\label{eq:gs-scaling}
\| x^2 \psi_\alpha \|_2 \leq C \alpha^{-1/2}
\end{align}

\item[(b)] Let $\psi_\alpha \in \mathcal{M}_{\mathcal{E}_\alpha}$.  Then there exists  $\alpha_0 >0$ and $C>0$ (independent of $\alpha$) such that for all $\alpha \geq \alpha_0$ we have 
\begin{align}
\label{eq:approx-gse-E}
\dist_{L^2} \left( \Theta ( \psi_\alpha), \psi_{\rm osc} \right) \leq C \alpha^{-1/20} \; . 
\end{align}

\item[(c)] Let $\psi_\alpha \in \mathcal{M}_{\mathcal{E}_\alpha}$ such that 
\begin{align}
\dist_{L^2} ( \Theta( \psi_\alpha), \psio ) = \| \psi_\alpha - \psio \|_2 \; . 
\end{align}
Then there exists $\alpha_0 >0 $ and $C>0$ (independent of $\alpha$) such that for  sufficiently large $\alpha \geq \alpha_0$ we have 
\begin{align}
\| x^2 \left( \psi_\alpha -  \psio  \right) \|_2  \leq C \alpha^{-11/20}  \; . 
\end{align}

\item[(d)] Under the same assumptions as in part {\rm (c)} there exists $\alpha_0 >0 $ and $C_1,C_2,C_3 >0$ (independent of $\alpha$) such that for  sufficiently large $\alpha \geq \alpha_0$ we have 
\begin{align}
\label{eq:gs-H1-lu}
C_1 \alpha^{1/4} \leq \| \nabla \psi_\alpha \|_2 \leq C_2 \alpha^{1/4} \; 
\end{align}
and furthermore 
\begin{align}
\|  \nabla \left( \psi_\alpha -  \psio  \right) \|_2  \leq C_3 \alpha^{9/40}  \; . 
\end{align}
\end{enumerate}

\end{lem}

\begin{proof}
First we remark that in the following proof we denote with $C>0$ a constant independent of $\alpha$. 

\textbf{Proof of  {\it (a)}:} For the upper bound of the ground state $e_\alpha$ we pick the harmonic oscillator's ground state $\psio$ defined in \eqref{def:min-osc} as trial state. Its energy serves as an upper bound for the ground state energy
\begin{align}
e_\alpha \leq \mathcal{E}_\alpha \left( \psio \right)  \; 
\end{align}
and can be explicitly computed with by the potential's \ref{eq:taylor-h} Taylor expansion, Ass. \ref{ass:reg} and $\cos( x) \geq 1 - \frac{x^2}{2}$ 
\begin{align}
\label{eq:upperbound-2}
e_\alpha \leq - \alpha \| g \|_2^2 + e_{\rm osc} + C
\end{align}
for a constant $C >0$ (independent of $\alpha$). 

For the lower bound we use the IMS localization technique to show that it suffices to consider the problem on a ball of radius $R$, where we can use the potential's Taylor expansion. To this end, let $\psi_\alpha$ denote a minimizer realizing
\begin{align}
e_\alpha = \inf_{\psi \in H^1( \mathbb{R}^3)} \mathcal{E}_\alpha ( \psi ) = \mathcal{E}_\alpha ( \psi_\alpha) \; 
\end{align}
and $\chi \in C^\infty ( \mathbb{R}^3)$ a function with support on the ball $B_1$ with radius one such that $\| \chi \|=1$ and $\chi (0) =1$. We define the rescaled function $\chi^R$ with $\| \chi^R \|_2 =1 $ supported on $B_R$ and denote with $\chi^{R,z} = \chi^R ( \cdot - z)$ its shift. The idea is to choose $R$ dependent on $\alpha$. However for simplicity we neglect the dependence of $R$ on $\alpha$ in the notation. We observe that the $L^2$-normalized function 
\begin{align}
\label{eq:def-psiR}
\psi^{R,z}_\alpha = \chi^{R,z} \psi_\alpha / \| \chi^{R,z} \psi_\alpha \|_2
\end{align}
satisfies 
\begin{align}
\label{eq:psiR-norm}
\vert \psi_\alpha (x) \vert^2 = \int_{B_R} \vert \chi^{R,z} \psi_\alpha  \vert^2 (x) dz = \int_{B_R} \vert\psi^{R,z}_\alpha  (x) \vert^2 \|\chi^{R,z} \psi_\alpha \|_2^2 dz 
\end{align}
and thus, by completing the square and standard techniques of IMS localization 
\begin{align}
\mathcal{E}_\alpha ( \psi_\alpha ) + \| \nabla \chi^{R}\|_2^2 = \int_{B_R} \left( \mathcal{E}_\alpha ( \psi^{R,z}_\alpha ) + \alpha \| v \varepsilon^{-1/2} (  \widehat{\varrho}_{\psi_\alpha} - \widehat{\varrho}_{\psi^{R,z}_\alpha} ) \|_2^2 \right) \|\chi^{R,z} \psi_\alpha \|_2^2 dz \; . 
\end{align}
Since $\| \psi^{R,z} \|_2^2 dz $ denotes a probability measure, there exists $z \in B_R$ such that 
\begin{align}
\label{eq:bound-R-0}
\mathcal{E}_\alpha ( \psi_\alpha) + \| \nabla \chi^{R}\|_2^2 \geq  \mathcal{E}_\alpha ( \psi^{R,z}_\alpha ) + \alpha \| v \varepsilon^{-1/2} (  \widehat{\varrho}_{\psi_\alpha} - \widehat{\varrho}_{\psi^{R,z}_\alpha} ) \|_2^2 \; . 
\end{align}
By scaling, we furthermore find $\| \nabla \chi^{R} \|_2^2 \leq C R^{-2}$ yielding 
\begin{align}
\label{eq:bound-R-1}
\mathcal{E}_\alpha ( \psi_\alpha)   \geq \mathcal{E}_\alpha ( \psi^{R,z}_\alpha ) + \alpha \| v \varepsilon^{-1/2} (  \widehat{\varrho}_{\psi_\alpha} - \widehat{\varrho}_{\psi^{R,z}_\alpha} ) \|_2^2 - CR^{-2} \; . 
\end{align}
We use \eqref{eq:bound-R-1} to prove both, the energy's lower bound and the approximation of the ground state. We start with the lower bound on the ground state energy first. For this,  we observe that \eqref{eq:bound-R-1} implies 
\begin{align}
\mathcal{E}_\alpha ( \psi_\alpha)   \geq \mathcal{E}_\alpha ( \psi^{R,z}_\alpha ) - CR^{-2} \;  , 
\end{align}
i.e. it suffices to compute the energy $\mathcal{E}_\alpha$ for function $\psi^{R,z}_\alpha$ supported on $B_R$, where we can use the Taylor expansion of $\ch$
\begin{align}
\mathcal{E}_\alpha ( \psi_\alpha^{R,z} ) &= \expval{- \frac{\Delta}{2m} + \alpha ( h* \vert \psi_{\alpha}^{R,z} \vert^2 )}{ \psi_\alpha^{R,z}}\notag \\
&\geq  - \alpha \| g \|_2^2 + \expval{- \frac{\Delta}{2m} + \alpha \| \nabla g \|_2^2 (  x^2 * \vert \psi_\alpha^{R,z} \vert^2 )}{ \psi_\alpha^{R,z}} - C \alpha R^4  \; . 
\end{align}
We observe that \eqref{eq:bound-R-0} is invariant w.r.t. translations and changes of phase of $\psi_\alpha$ and $\psi^{R,z}$ and thus, we can furthermore restrict to $\psi_\alpha^{R,z}$ such that 
\begin{align}
\label{eq:psiR-mean}
\expval{x}{\psi_\alpha^{R,z}} =0
\end{align}
for which we find 
\begin{align}
\label{eq:upperbound-1}
\mathcal{E}_\alpha ( \psi_\alpha)   \geq \expval{{\rm h}_{\rm osc}}{ \psi^{R,z}_\alpha } - \alpha \| g \|_2^2- C \alpha R^4 - CR^{-2} \; . 
\end{align}
where $\ho$ denotes the harmonic oscillator ${\rm h}_{\rm osc} = - \frac{\Delta}{2m} + \frac{m\omega^2}{2} x^2$. By definition, $\psi^{R,z}_\alpha$ is $L^2$-normalized and thus a competitor for the ground state of $\ho$, i.e. 
\begin{align}
\mathcal{E}_\alpha ( \psi_\alpha)   &\geq \inf_{\psi \in H^1}\expval{{\rm h}_{\rm osc}}{ \psi } - \alpha \| g \|_2^2 - C \alpha R^4 - CR^{-2} \notag \\
&\geq e_{\rm osc} - \alpha \| g \|_2^2 - C \alpha R^4 - CR^{-2} \; .
\end{align}
Optimizing w.r.t. to the parameter $R$ (yielding $R= \alpha^{-1/6}$), we arrive at
\begin{align}
\label{eq:upper}
e_\alpha  \geq e_{\rm osc} - \alpha \| g \|_2^2 - C \alpha^{1/3} \;   
\end{align}
proving part {\em (a)}.

\textit{Properties of $\psi_\alpha, \psi_\alpha^{R,z}$:} As a preliminary step to prove the remaining parts of this Proposition we prove useful properties of the ground state $\psi_\alpha$ and $\psi_\alpha^{R,z}$ (constructed in \eqref{eq:def-psiR}, satisfying \eqref{eq:bound-R-0} and by translational invariance of the problem \eqref{eq:psiR-mean}). We observe that $\cos ( x ) \leq 1$ (and thus $h \leq \| g \|_2^2$) implies 
\begin{align}
\expval{- \tfrac{\Delta}{2m}}{ \psi_\alpha^{R,z}} = \mathcal{E}_\alpha (\psi_\alpha^{R,z} ) +   \alpha \expval{(h* \vert \psi_{\alpha}^{R,z}\vert^2)}{\psi_\alpha^{R,z}}  \leq \mathcal{E}_\alpha (\psi_\alpha^{R,z} )  + \alpha \| g \|_2^2 \; .  \label{eq:Delta}
\end{align}
With\eqref{eq:bound-R-1} and the ground state energy's approximation (part {\it (a)}) we obtain 
\begin{align}
\| \psi_\alpha^{R,z} \|_{H^1}^2 \leq C \sqrt{\alpha} , \quad \text{and similarly} \quad \| \psi_\alpha \|_{H^1}^2 \leq C \sqrt{\alpha} \; , 
\end{align}
and in the same way
\begin{align}
\label{eq:bound-varrho}
\alpha \| v/\varepsilon^{1/2} \widehat{\varrho}_{\psi_\alpha} \|_2^2, \; \alpha \| v/\varepsilon^{1/2} \widehat{\varrho}_{\psi_\alpha^{R,z}} \|_2^2  \leq C \alpha \; . 
\end{align}
With the upper bound \eqref{eq:upper} we furthermore deduce from \eqref{eq:Delta} resp. the Euler-Lagrange equation 
\begin{align}
\| \Delta \psi_\alpha \|_2^2 \leq C \sqrt{\alpha} \label{eq:gs-H2}
\end{align}
for all $\alpha \geq \alpha_0$.  However for $\psi_{\alpha}^{R,z}$ we observe first that \eqref{eq:bound-R-1} resp. \eqref{eq:upperbound-1} together with the harmonic oscillator's coercivity property \eqref{eq:coerc-osc} show for $R= \alpha^{-1/6}$  
\begin{align}
\sqrt{\alpha} \inf_{\theta}\| \psi_\alpha^{R,z} - e^{i \theta}\psio \|_2^2 &\leq  \expval{{\rm h}_{\rm osc}}{ \psi_\alpha^{R,z} } - e_{\rm osc} \notag \\ 
&\leq \mathcal{E}_\alpha ( \psi_\alpha^{R,z} ) - e_{\rm osc}+ \alpha \| g \|_2^2 + C \alpha^{1/3} \; . 
\end{align}
With the upper bound on the energy \eqref{eq:upperbound-2} we find 
\begin{align}
\inf_{\theta}\| \psi^{R,z}_\alpha - e^{i \theta} \psio \|_2^2  \leq C \alpha^{-1/6}, \quad \inf_{\theta}\| \psi^{R,z}_\alpha - e^{i \theta} \psio \|_{H^1}^2  \leq C \alpha^{1/12} \;  \,  . \label{eq:approx-1}
\end{align}
In particular for sufficiently large $\alpha \geq \alpha_0$ we have 
\begin{align}
 \|  \nabla \psi_\alpha^{R,z} \|_2 \leq C \alpha^{-1/4} \; . 
\end{align}

\textit{Approximation of Lagrange multipliers:} We prove the approximation of the Lagrange multiplier $\mu_{\psi_\alpha}$ with $\mu_{\psio}$ using the previous results. In particular by translational invariance of the problem we choose $\psi_\alpha$ such that $\psi_\alpha^{R,z}$ (as constructed in \eqref{eq:def-psiR}) satisfies \eqref{eq:bound-R-0} and \eqref{eq:psiR-mean}). We write 
\begin{align}
\mu_{\psi_\alpha} - \mu_{\psio} = \mathcal{E}_\alpha ( \psi_\alpha) - \mathcal{E}_\alpha ( \psio ) - \alpha \| v / \varepsilon^{1/2} \varrho_{\psi_\alpha} \|_2^2 + \alpha \| v / \varepsilon^{1/2} \varrho_{\psio} \|_2^2 
\end{align}
yielding with \eqref{eq:bound-R-0} and \eqref{eq:bound-varrho} to 
\begin{align}
\vert \mu_{\psi_\alpha} - \mu_{\psio} \vert \leq& \vert  \mathcal{E}_\alpha ( \psi_\alpha) - \mathcal{E}_\alpha ( \psio )  \vert  \notag \\
&+\alpha \left( \| v / \varepsilon^{1/2} \varrho_{\psio} \|_2 +  \| v / \varepsilon^{1/2} \varrho_{\psi_\alpha} \|_2 \right) \| v / \varepsilon^{1/2} \left(  \varrho_{\psi_\alpha}  -\varrho_{\psio} \right) \|_2  \notag \\
\leq& C \alpha^{1/3} + \sqrt{\alpha}\| v / \varepsilon^{1/2} \left(  \varrho_{\psi_\alpha}  -\varrho_{\psio} \right) \|_2 \; . \label{eq:mu-10}
\end{align}
Since 
\begin{align}
\| v / \varepsilon^{1/2} (  \varrho_{\psi_\alpha}  -\varrho_{\psio} ) \|_2 \leq \| v / \varepsilon^{1/2} (\varrho_{\psi_\alpha^{R,z}}  -\varrho_{\psio} ) \|_2 + \| v / \varepsilon^{1/2} (  \varrho_{\psi_\alpha^{R,z}}  -\varrho_{\psio} ) \|_2
\end{align}
we find with \eqref{eq:bound-R-0}, $\| v / \varepsilon^{1/2} ( \varrho_{\psi_\alpha^{R,z}}  -\varrho_{\psio}) \|_2 \leq C \| \psio - \psi_\alpha^{R,z}\|_2$ and \eqref{eq:approx-1} 
\begin{align}
\| v / \varepsilon^{1/2} \left(  \varrho_{\psi_\alpha}  -\varrho_{\psio} \right) \|_2 \leq C \alpha^{-1/3} + C \alpha^{-1/12} \; . 
\end{align}
Thus we obtain from \eqref{eq:mu-10} for sufficiently large $\alpha >\alpha_0$ 
 \begin{align}
 \label{eq:mu-bound}
 \vert \mu_{\psi_\alpha} - \mu_{\psio} \vert \leq& C  \alpha^{5/12} \; . 
 \end{align}
\textit{Scaling of the ground state:} 
To show the ground state's scaling \eqref{eq:gs-scaling} for $\psi_\alpha$ satisfying \eqref{eq:ass-gs} we observe that by the Euler-Lagrange equation we have 
\begin{align}
\label{eq:weighted-dist-0}
\left( - \Delta - \mu_{\psi_\alpha} \right)\psi_\alpha = 2 \alpha \left( h * \vert \psi_\alpha \vert^2 \right) \psi_\alpha \; . 
\end{align}
The idea is now to use the properties of the resolvent $\left( - \Delta  - \mu_{\psi_\alpha} \right)^{-1}$ (that is well defined sind from \eqref{eq:mu-bound} we have $\mu_{\psi_\alpha} \geq - C \alpha $ for sufficiently large $\alpha \geq \alpha_0$) to prove the desired bound. The resolvent's Green's function is given in terms of the inverse Fourier transform $\cF^{-1}$ of 
\begin{align}
\label{eq:Green1}
G_{\mu_{\psi_\alpha}} (z) = \cF^{-1} \left[ \big ( p^2 /(2m) -  \mu_{\psi_\alpha} \big)^{-1} \right] (z)
\end{align}
and can by functional calculus explicitly computed. In fact we have 
\begin{align}
\mathcal{F}\left[ \left(- \tfrac{1}{2m}\Delta -   \mu_{\psi_\alpha} \right)^{-1} \right] (p) = \int_0^\infty e^{- t (\tfrac{1 }{2m}  p^2 - \mu_{\psi_\alpha} )  } dt 
\end{align}
which leads with \eqref{eq:Green1} to 
\begin{align}
G_{ \mu_{\psi_\alpha}} (z) &= \frac{1}{\left( 2 \pi \right)^{3/2}} \int_{\bR^3} \int_0^\infty e^{i z \cdot p} e^{t  \mu_{\psi_\alpha}} e^{-\tfrac{t}{2m}  p^2} dt dp 
\end{align}
and we arrive with Fubini's theorem at 
\begin{align}
G_{\mu_{\psi_\alpha}} (z) =  (2m)^{3/2} \int_0^\infty e^{t  \mu_{\psi_\alpha} } \frac{e^{-\frac{mz^2}{2t}}}{ t^{3/2} } dt  = (2m)^{3/2}  \frac{\sqrt{\pi}}{\vert z \vert } e^{-\sqrt{-4m \mu_{\psi_\alpha}} \vert z \vert}
\end{align}
for $\vert z \vert \not=0$.  We plug this identity into \eqref{eq:weighted-dist-0} and find using assumption \eqref{eq:ass-gs}, and the notation $F: = 2 \alpha \left( h * \vert \psi_\alpha \vert^2 \right) \psi_\alpha $ (i.e. $F$ denotes the right hand side of \eqref{eq:weighted-dist-0}) 
\begin{align}
\psi_\alpha (x)  =   \int G_{\mu_{\psi_\alpha}} (x-y) \; F(y) dy \; .
\end{align}
Thus the weighted $L^2$-norm,  we aim to find an upper bound for, becomes 
\begin{align}
\| x^{3} \psi_\alpha \|_2^2 
=&   \int x^{6}  G_{\mu_{\psi_\alpha}} (x-y)G_{\mu_{\psi_\alpha}} (x-z)   \; F(z) F(y) \; dxdydz 
\end{align}
that we can estimate by Cauchy Schwarz with 
\begin{align}
\| x^{3}\psi_\alpha \|_2^2 
\leq&  \int x^{6} \; \vert G_{\mu_{\psi_\alpha}} (x-y) \vert^2  \vert F(y) \vert^2 dxdy \notag\\
=&  (2m)^3 \pi\int \frac{ x^{6}}{ \vert x -y \vert^2 } e^{-2 \sqrt{-4m \mu_{\psi_\alpha} } \vert x-y \vert } \vert  F(y) \vert^2 dydx \; . 
\end{align}
We split the integral into several regions to find the desired bound: First we consider the case $\vert x \vert > 2 \vert y \vert $  for which we have  $\vert x - y\vert \geq \vert x \vert - \vert y \vert \geq \frac{\vert x \vert }{2}$. By substitution we can therefore estimate the integral in this region by 
\begin{align}
&  \int_{\vert x \vert > 2 \vert y \vert }  \frac{  x^{6}}{\vert x- y\vert^2} \; e^{-2 \sqrt{-4m \mu_{\psi_\alpha} } \vert x-y \vert } \vert F(y) \vert^2 dydx \notag\\
 &\quad \leq   \int \; x^{4} e^{-2 \sqrt{-4m \mu_{\psi_\alpha} } \vert x \vert }  \vert F(y) \vert^2 dydx \leq \frac{C }{ \alpha^{5/2}} \| F \|_2^2 
\end{align}
where we used $- \mu_{\psi_\alpha} \geq C \alpha$ for sufficiently large $\alpha \geq \alpha_0$. Now let $\vert x \vert <2  \vert y \vert $ and $\vert x-y\vert > \frac{\vert y \vert}{2}$, then we have $\frac{\vert x \vert }{\vert x-y\vert } \leq 4$ and it follows 
\begin{align}
 \int_{\substack{\vert x \vert <  2 \vert y \vert \\ \vert x-y\vert \geq \vert y \vert /2}} &\frac{  x^{6} }{\vert x- y\vert^2}\;e^{-2 \sqrt{-4m \mu_{\psi_\alpha} } \vert x-y \vert } \vert F(y) \vert^2 dydx \notag\\
 &\leq 4  \int x^{4} e^{- \sqrt{-4m \mu_{\psi_\alpha} } \vert y \vert } \vert F(y) \vert^2  dydx \leq \frac{C }{ \alpha^{5/2}}\| F \|_2^2 \; . 
\end{align}
with similar arguments as before. Finally for $\vert x \vert < 2 \vert y \vert $ and $\vert x-y\vert \leq \frac{\vert y \vert}{2}$ it follows $\vert x \vert \geq \vert y \vert - \vert x - y \vert \geq \vert y \vert/2$. Hence 
\begin{align}
 \int_{\substack{\vert x \vert <  2 \vert y \vert \\ \vert x-y\vert < \vert y \vert /2}} & \frac{  x^{6}}{\vert x- y\vert^2} \; e^{-2 \sqrt{-4m \mu_{\psi_\alpha} } \vert x-y \vert } \vert F(y) \vert^2 dydx \notag\\
 &\leq 2  \int_{ \vert y \vert /2<  \vert  x \vert <  2 \vert y \vert } \frac{ y^{6}}{\vert x-y\vert^2}  e^{- \sqrt{-4m \mu_{\psi_\alpha} } \vert x- y \vert } \vert F(y) \vert^2 dydx \notag\\
 &\leq \frac{C}{( -\mu_{\psi_\alpha} )^{1/2}}  \int y^{6}   e^{-  C\sqrt{-4m \mu_{\psi_\alpha} } \vert y \vert } \vert F(y)  \vert^2 dy \leq \frac{C }{ \alpha^{5/2}}  \| F\|_2^2 
\end{align}
where we used that $ \vert y \vert^{6} \;   e^{- C \sqrt{-4m  \mu_{\psi_\alpha}} \vert y\vert}   \leq \frac{\widetilde{C}}{ \alpha^{3}}$. Summarizing the estimates we conclude by 
\begin{align}
\label{eq:estimate-F}
\| x^3 \psi_\alpha \|_2 \leq C \alpha^{-5/4 }\| F \|_2 \;
\end{align}
where $F$ denotes the r.h.s. of \eqref{eq:weighted-dist-0}. Since $\| h \|_\infty \leq C$ ,$\psi_\alpha$ is a $L^2$-normalized function we find that
\begin{align}
\| x^{3} \psi_\alpha \|_2 \leq C \alpha^{-3/4} \; . 
\end{align}  
for sufficiently large $\alpha \geq \alpha_0$. In particular for $n=3$ we find $\|x^3 \psi_\alpha \|_2 \leq C \alpha^{-3/4}$ and thus, in particular, 
\begin{align}
\| x^2 \psi_\alpha \|_2^2 \leq C \alpha^{-1/2} \; . 
\end{align}

\textbf{Proof of {\it (b)}:} In order to prove the ground state $\psi_\alpha$'s approximation we observe that by the previous discussion it is enough to consider the problem on the ball $B_{R_1} (0) $ with $R_1 = \alpha^{-1/5}$. In fact \eqref{eq:gs-scaling} shows 
\begin{align}
\label{eq:apprxo-step1}
\inf_{\theta} \| \psio - e^{i \theta} \psi_\alpha \|_2 \leq \| \psio - e^{i \theta} \psi_\alpha \|_{ L^2( B_{R_1} (0)) }  + C \alpha^{-1/20} \; . 
\end{align}
We consider $\psi_\alpha$ and $\psi_\alpha^{R,z}$ (constructed in \eqref{eq:def-psiR}), satisfying \eqref{eq:bound-R-0} and by translational invariance of the problem \eqref{eq:psiR-mean}). With these notations we in particular have from \eqref{eq:approx-1}  
\begin{align}
\inf_{\theta} \| \psio - e^{i \theta} \psi_\alpha \|_{ L^2( B_{R_1} (0)) } &\leq   \| \psi_\alpha - \psi_\alpha^{R,z} \|_{ L^2( B_{R_1} (0)) } + \inf_{\theta} \| \psio - e^{i \theta} \psi_\alpha^{R,z} \|_{ L^2( B_{R_1} (0)) } \notag \\
 &\leq  \| \psi_\alpha - \psi_\alpha^{R,z} \|_{ L^2( B_{R_1} (0)) } + C \alpha^{-1/12} 
\end{align}
and thus it remains to show $\psi_\alpha^{R,z}$ is close to $\psi_\alpha$. To this end we control the $L^2$-norm of the difference $\varrho_{\psi_\alpha} -  \varrho_{\psi^{R,z}_\alpha}$ first and deduce as a second step from this estimates for $L^2$-norm of $\psi_\alpha^{R,z} -\psi_\alpha$. For this we write the $L^2$-norm of $\varrho_{\psi_\alpha} -  \varrho_{\psi^{R,z}_\alpha}$ in momentum space. We shall show that for low momenta, we control the norm with Ass. \ref{ass:reg} and \eqref{eq:bound-R-0} while for high momenta the $L^2$-norm is small by regularity properties of $\varrho_{\psi_\alpha},\varrho_{\psi^{R,z}_\alpha}$. For the latter one we observe that from \eqref{eq:gs-H2} we have 
\begin{align}
\| k^2 \widehat{\varrho}_{\psi_\alpha} \|_\infty \leq C \| \Delta \psi_\alpha \|_2 \| \psi_\alpha \|_2 + C \| \nabla \psi_\alpha \|_2^2 \leq C \alpha^{1/2} \; 
\end{align}
and thus there exists $C>0$ such that 
\begin{align}
\label{eq:bound-widehatrho1}
\vert \widehat{\varrho}_{\psi_\alpha}  (k) \vert \leq \frac{C \alpha^{1/2}}{\vert k \vert^2} \, , 
\end{align}
To obtain a similar bound for $\widehat{\varrho}_{\psi_\alpha^{R,z}}$ we first have to derive a bound for the $H^2$-norm of $\psi_\alpha^{R,z}$. For this we remark that by definition \eqref{eq:def-psiR} we have
\begin{align}
\| \chi^{R,z} \psi_\alpha \|_2 \Delta \psi_\alpha^{R,z} = \chi^{R,z} ( \Delta \psi_\alpha ) + 2 ( \nabla \chi^{R,z}) ( \nabla \psi_\alpha) + ( \Delta \chi^{R,z} ) \psi_\alpha \; . 
\end{align}
With the Taylor expansion of $\chi^{R,z}$ around $z$ we find that $\| \chi^{R,z} \psi_\alpha \|_2  \geq \| \psi_\alpha \|_2 + R^{-1} \| x ( \nabla \chi^{R,z} ( w) ) \psi_\alpha \|_2 $ for $w \in B_R (z)$ and thus with the ground state's scaling properties \eqref{eq:gs-scaling} we have $\| \chi^{R,z} \psi_\alpha \|_2  \geq C $ for sufficiently large $\alpha \geq \alpha_0$. Hence we find with \eqref{eq:gs-H1-lu}, \eqref{eq:gs-H2} and the scaling of $\chi^{R,z}$ at $\| \Delta \psi^{R,z}_\alpha \|_2 \leq C \sqrt{\alpha}$ and thus 
\begin{align}
\label{eq:bound-widehatrho2}
\vert \widehat{\varrho}_{\psi_\alpha^{R,z}}  (k) \vert \leq \frac{C \alpha^{1/2}}{\vert k \vert^2}   \; . 
\end{align} 
From \eqref{eq:bound-widehatrho1} and \eqref{eq:bound-widehatrho2} we find that for high momenta, i.e. all $k \in B_{R_2}^c = \lbrace k \in \mathbb{R}^3 : \vert k \vert > R_2 \rbrace$ we have 
\begin{align}
\| \widehat{\varrho}_{\psi_\alpha} -  \widehat{\varrho}_{\psi^{R,z}_\alpha} \|_{L^2( B_{R_2}^c )} \leq \frac{C\alpha^{1/2}}{R_2^{1/2}} \; .
\end{align}
We remark that we choose $R_2 := R_2 ( \alpha) = \alpha^{1/5}$, however, for simplicity (as for $R$), we neglect the dependence of $\alpha$ in its notation. 
For low momenta, i.e. $k \in B_{R_2}$ we use that by Ass. \eqref{ass:reg} we have $\vert \widehat{g} (k) \vert \geq ( 1 + \vert k \vert )^{-9/2}$ and thus  
\begin{align}
\|  \widehat{\varrho}_{\psi_\alpha} - \widehat{\varrho}_{\psi^{R,z}_\alpha} \|_{L^2( B_{R_2})}  \leq&  C ( 1 + \vert R_2 \vert )^{9/2}   \| v \ \varepsilon^{-1/2}  ( \widehat{\varrho}_{\psi_\alpha} - \widehat{\varrho}_{\psi^{R,z}_\alpha} ) \|_{L^2( B_{R_2})} \notag \\
\leq& C ( 1 + \vert R_2 \vert )^{9/2}   \| v \ \varepsilon^{-1/2}  ( \widehat{\varrho}_{\psi_\alpha} - \widehat{\varrho}_{\psi^{R,z}_\alpha} ) \|_{2}  \; 
\end{align}
and we find with \eqref{eq:bound-R-0} (assuming $R_2 \leq \alpha^{1/5}$) 
\begin{align}
\|  \widehat{\varrho}_{\psi_\alpha} - \widehat{\varrho}_{\psi^{R,z}_\alpha} \|_2 \leq&  C  \alpha^{-1/3}  ( 1 + \vert R_2 \vert )^{9/2}  \; . \label{eq:estimate-varrho-H12}
\end{align}
Optimizing with respect to $R_2$ leads to $R_2 = \alpha^{1/5}$ and thus for sufficiently large $\alpha \geq \alpha_0$ to 
\begin{align}
\|  \widehat{\varrho}_{\psi_\alpha} - \widehat{\varrho}_{\psi^{R,z}_\alpha} \|_2 \leq&  C  \alpha^{2/5}  \; . 
\end{align}
In particular it follows as $\psi_\alpha^{R,z}$ and $\psi_\alpha$ have the same phase 
\begin{align}
\| \psi_\alpha - \psi_\alpha^{R,z} \|_4^4 =\| \vert \psi_\alpha \vert -\vert  \psi_\alpha^{R,z} \vert\|_4^4 \leq \|  \widehat{\varrho}_{\psi_\alpha} - \widehat{\varrho}_{\psi^{R,z}_\alpha} \|_2^2 \leq C \alpha^{4/5} . 
\end{align}
We recall that we need to estimate the $L^2$-norm difference of $\psi_\alpha, \psi_\alpha^{R,z}$ on the ball $B_{R_1}$,  i.e. we have by Cauchy Schwarz's inequality with $R_1= \alpha^{-1/5}$ 
\begin{align} 
\|  \psi_\alpha - \psi_\alpha^{R,z} \|_{L^2 ( B_{R_1})}^2 \leq C R_1^{3/2} \alpha^{2/5}  \leq C \alpha^{-1/5} \;  \label{eq:approx-3}
\end{align}
and we arrive with \eqref{eq:apprxo-step1} at 
\begin{align}
\inf_\theta \| e^{i \theta} \psi_\alpha - \psio \|_{2}^2 \leq C \alpha^{-1/10}  \; . 
\end{align}

\textbf{Proof of part {\it (c)}:} In order to prove $H^1$-norm convergence of $\psi_\alpha$ to $\psio$ we observe that from the Euler-Lagrange equations of $\psi_\alpha$ resp. $\psio$, Ass.\ref{ass:reg} and the scaling properties of $\psio$
\begin{align}
- \tfrac{1}{2m} \Delta    \left( \psi_\alpha - \psio \right) =&  \left( \alpha   \left( h *   \vert \psi_\alpha \vert^2 \right)  + \mu_{\psi_\alpha} \right) \psi_\alpha  - \left( \alpha   \left( h *  \vert \psio  \vert^2 \right)  +  \mu_{\psio}  \right) \psio  \notag \\
&+ C  \left( \alpha ( \vert \cdot \vert^4 * \vert \psio \vert^2 )  + 1 \right)\psio  
\end{align}
for a constant $C>0$ independent of $\alpha$. We recall that both $\psi_\alpha$ and $\psio$ are $L^2$-normalized functions. Hence introducing the notation $\widetilde{h} = h - \| g \|_2^2$ and 
\begin{align}
\widetilde{\mu}_\psi = \expval{-\tfrac{\Delta}{2m} - \alpha ( \widetilde{h} * \vert \psi \vert^2 )}{\psi}
\end{align}
for any $\psi \in H^1 ( \mathbb{R}^3)$ we arrive at 
\begin{align}
- \tfrac{1}{2m} \Delta    \left( \psi_\alpha - \psio \right) =& \alpha   \left( \widetilde{h} *  \psi_\alpha \overline{\left( \psi_\alpha - \psio\right)}  \right) \psi_\alpha + \alpha  \left( \widetilde{h} * \left( \psi_\alpha - \psio \right)  \overline{ \psio}  \right) \psi_\alpha \notag \\
&+ \alpha  \left( \widetilde{h} *\vert \psio \vert^2  \right) \left( \psio - \psi_\alpha \right) - \left( \widetilde{\mu}_{\psi_\alpha} - \widetilde{\mu}_{\psio} \right) \psi_\alpha  + \widetilde{\mu}_{\psi_\alpha} \left( \psio - \psi_\alpha \right) \notag \\
&+ \widetilde{\mu}_{\psio} \left( \psi_\alpha - \psio \right) + C \alpha ( \vert \cdot \vert^4 * \vert \psio \vert^2 ) \psio   
\end{align}
From Ass. \ref{ass:reg} we have $\vert \widetilde{h} (x) \vert \leq C x^2$ and thus together with the ground state's scaling properties (part {\it (c)}) we find $\vert \widetilde{\mu}_{\psi_\alpha} \vert \leq C \alpha^{1/2}$, $\vert \widetilde{\mu}_{\psio} \vert \leq C \alpha^{1/2}$ and 
\begin{align}
\vert  \widetilde{\mu}_{\psi_\alpha} - \widetilde{\mu}_{\psio} \vert = \vert  \mu_{\psi_\alpha} - \mu_{\psio} \vert \leq C \alpha^{5/12} 
\end{align}
from part (part {\it (a)}). Therefore we find with the approximation of the ground state (part {\it (b)}) and the ground state's scaling properties (part {\it (c)}) that 
\begin{align}
\vert \bra{\psio}\ket{- \Delta \left( \psi_\alpha - \psio \right)} \vert, \vert \bra{\psi_\alpha}\ket{- \Delta \left( \psi_\alpha - \psio \right)} \vert \leq C \alpha^{9/20} 
\end{align} 
finally yielding 
\begin{align}
\| \nabla \psi_\alpha - \nabla \psio \|_2^2 \leq C \alpha^{9/40} 
\end{align}
and the desired lower and upper bound in \eqref{eq:gs-H1-lu}. 

\textbf{Proof of {\it (d)}:} To prove convergence of $\psi_\alpha$ to $\psio$ for $\psi_\alpha$ satisfying \eqref{eq:ass-gs} in the weighted $L^2$-norm, too, we proceed similarly as in part {\it (d)}. From the Euler-Lagrange equation of $\psi_\alpha$ we get 
\begin{align}
& \left(\tfrac{1}{2m} \Delta - \mu_{\psi_\alpha} \right)  \left( \psi_\alpha - \psio \right)\notag  \\ 
&\quad =2 \alpha  \left( h * \vert \psi_\alpha \vert^2 \right)  \psi_\alpha - 2\alpha  \left( \left( h * \vert \psio \vert^2 \right) + ( \mu_{\psio} - \mu_{\psi_\alpha}) \right) \psio \notag \\
&\quad = 2 \alpha   \left( h *  \psi_\alpha \overline{\left( \psi_\alpha - \psio\right)}  \right) \psi_\alpha +2 \alpha  \left( h * \left( \psi_\alpha - \psio \right)  \overline{ \psio}  \right) \psi_\alpha \notag \\
&\quad \quad + 2 \alpha  \left( h *\vert \psio \vert^2  \right) \left( \psio - \psi_\alpha \right) - \left( \mu_{\psi_\alpha} - \mu_{\psio} \right) \psio + C \alpha ( \vert \cdot \vert^4 * \vert \psio \vert^2 ) \psio   \; .  \label{eq:weighted-dist-1}
\end{align}
By assumption resp. part {\it (a)} we have
\begin{align}
\dist_{L^2} \left( \psio, \Theta (\psi_\alpha ) \right) = \| \psi_\alpha - \psio \|_2 \leq C \alpha^{-1/20}  
\end{align}
and furthermore we have from \eqref{eq:estimate-F} (using assumption \eqref{eq:ass-gs}) the estimate 
\begin{align}
\| x^2 (\psi_\alpha - \psio ) \|_2 \leq C \alpha^{-1/2} \| F \|_2  \; . 
\end{align}
where $F$ denotes the r.h.s. of \eqref{eq:weighted-dist-1}.  It follows from the approximation of the Lagrange multiplier (part {\it (a)}) and the ground state  (part {\it (b)}) that $\| F \|_2 \leq C \alpha^{-1/20}$ and we finally arrive at the desired bound of part {\it (c)}.

\end{proof}

\subsection{Properties of the Hessian} The Hessian $\mathcal{H}_{\alpha}$ of $\mathcal{E}_\alpha$ is defined for any $\psi_\alpha \in \mathcal{M}_{\cE_\alpha}$ by 
\begin{align}
\mathcal{H}_{\alpha}:= \lim_{\varepsilon \rightarrow 0} \frac{1}{\delta^2} \left( \cE_\alpha \left( \frac{\psi_\alpha + \delta f}{ \| \psi_\alpha + \delta f \|_2} \right) - e_\alpha \right) \quad \text{for all} \quad f \in H^1 ( \mathbb{R}^3 ) \; . 
\end{align}
We can explicitly compute the Hessian and find 
\begin{align}
\mathcal{H}_\alpha := \expval{{\rm H}_{\alpha}}{\Im f} + \expval{Q_\alpha \left(\rH - 4 X_{\psi_\alpha}\right)Q_\alpha}{\Re f}
\end{align}
where $Q_\alpha = 1 - P_\alpha = 1 - \ket{\psi_\alpha} \bra{\psi_\alpha}$ and 
\begin{align}
\label{def:X}
\rH = \rh_{\sqrt{\alpha} \varphi_\alpha} - \mu_{\psi_\alpha}, \quad \text{and} \quad  X_{\psi_\alpha} (x;y)  =   \psi_\alpha (x) \; \ch (x-y) \; \psi_\alpha (y) \; . 
\end{align}


\subsubsection*{Positivity of the Hessian}  We compare the Hessian's components 
\begin{align}
\mathcal{H}_\alpha^{(1)} &:= \inf_{\psi_\alpha \in \mathcal{M}_{\cE_\alpha}} \inf_{\substack{f \in H^1( \bR^3), \|f \|_2 =1\\f \in \left( \rm{span} \lbrace \psi_\alpha \rbrace \right)^\perp} } \expval{\rH}{f}, \notag\\
 \mathcal{H}_\alpha^{(2)} &:= \inf_{\psi_\alpha \in \mathcal{M}_{\mathcal{E}_\alpha}} \inf_{\substack{f \in H^1( \bR^3), \|f \|_2 =1\\f \in \left( \rm{span} \lbrace \psi_\alpha, \partial_1 \psi_\alpha, \partial_2 \psi_\alpha , \partial_3 \psi_\alpha  \rbrace \right)^\perp} } \expval{\rH - X_{\psi_\alpha}}{f}
\end{align}
with $\mathcal{H}_{\rm osc}$ (defined in \eqref{def:hessian-osc}) known to satisfy $\mathcal{H}_{\rm osc} \geq C \sqrt{\alpha}$ for a constant $C>0$ independent of $\alpha$. We remark that by definition the Hessian $\mathcal{H}_\alpha$ is defined modulo its zero modes namely the ground state $\psi_\alpha$ for the first component resp. $\psi_\alpha$ and its partial derivatives $\partial_j \psi_\alpha$ for $i=1,2,3$ for the second component.

\begin{lem}
\label{lemma:Hessian}
Let $\varepsilon, v$ satisfy Assumption \ref{ass:reg}. Then, there exists $\alpha_0>0$ and $C>0$ (independent of $\alpha$) such that 
\begin{align}
\mathcal{H}_\alpha^{(i)} \geq  C \alpha^{1/2}, \quad \text{for all} \quad \alpha \geq \alpha_0   \label{eq:Hessian-vgl} \; . 
\end{align}
\end{lem}

\begin{proof}
We present the proof of the Hessian's component $\mathcal{H}_\alpha^{(2)}$. The statement for $\mathcal{H}^{(1)}$ then follows with similar arguments. 

For any $\psi_\alpha \in \mathcal{M}_{\mathcal{E}_\alpha}$ and $f \in H^1 ( \mathbb{R}^3)$ we define the projection  $Q_\alpha' = 1- P_\alpha', P_\alpha' = \ket{\psi_\alpha} \bra{\psi_\alpha} + \sum_{j=1}^3 \ket{\partial_j \psi_\alpha} \bra{\partial_j \psi_\alpha} / \| \partial_j \psi_\alpha \|_2^2 $ and furthermore the $L^2$-normalized function 
\begin{align}
g_\alpha := \frac{Q_\alpha' f}{\| Q_\alpha' f \|_2} \; . 
\end{align}
Thus in the following we consider the expectation value
\begin{align}
 \expval{\rH - X_{\psi_\alpha}}{g_\alpha}  = \expval{\rH - \widetilde{X}_{\psi_\alpha} }{ g_\alpha } 
\end{align}
where we introduced the notation 
\begin{align}
\widetilde{X}_{\psi_\alpha} (x;y)=   \psi_\alpha (x) \; \widetilde{h} (x-y) \; \psi_\alpha (y) \; 
\end{align}
with $\widetilde{h} = h - \| g \|_2^2$ and used that the zero-th order term of the expansion of $\ch$ in the above expectation value vanishes as $Q'_\alpha f$ is orthogonal to $\psi_\alpha$. 
By translational invariance of the problem we restrict to $\psi_\alpha$ such that $\dist_{L^2} ( \Theta (\psi_\alpha ), \psio) = \| \psio - \psi_\alpha \|_2$ (so that Lemma \ref{lemma:approx-E} {\it (c), (d)} apply). We recall that we want to compare $\mathcal{H}_\alpha^{(i)}$ with $\mathcal{H}_{\rm osc}$ that is of order $\sqrt{\alpha}$. Thus any term we can show to be $o( \sqrt{\alpha})$ will be dominated in the end by $\mathcal{H}_{\rm osc}$ and thus will be considered to be sub-leading for sufficiently large $\alpha \geq \alpha_0$. 

We shall first show that $g_\alpha$ is approximately orthogonal to the harmonic oscillator's ground and first excited state $P_{\rm osc}' =  \ket{\psio} \bra{\psio} + \sum_{j=1}^3 \ket{\partial_j \psio} \bra{\partial_j \psio} / \| \psio \|_2^2 $, i.e. that it is enough to consider 
\begin{align}
\label{eq:Qosc}
\expval{ {{\rm H}}_\alpha - \widetilde{X}_{\psi_\alpha} }{ Q_{\rm osc}' g_\alpha } \; . 
\end{align}
This follows from the observation that the difference is given by 
\begin{align}
 & \expval{{{\rm H}}_\alpha - \widetilde{X}_{\psi_\alpha} }{  g_\alpha} - \expval{ {{\rm H}}_\alpha - \widetilde{X}_{\psi_\alpha} }{ Q_{\rm osc}' g_\alpha } \notag \\
  & \quad =  2 \Re \bra{P_{\rm osc}'g_\alpha }{{\rm H}}_\alpha -\widetilde{X}_{\psi_\alpha} \ket{  g_\alpha  } + \bra{P_{\rm osc}'g_\alpha }{{\rm H}}_\alpha -\widetilde{X}_{\psi_\alpha} \ket{ P_{\rm osc}'   g_\alpha } \; .
\end{align}
On the one hand, since 
\begin{align}
\| P_{\rm osc} Q_\alpha g_\alpha \|_2 \leq C \| \psi_\alpha - \psio \|_2 \| g_\alpha \|_2 \leq C \alpha^{-1/20} \| g_\alpha \|_2
\end{align}
and similarly denoting $P_{\rm osc}^{(j)} =: \ket{\partial_j \psio} \bra{\partial_j \psio} / \| \psio \|_2^2 $ and $P_{\alpha}^{(j)} =: \ket{\partial_j \psi_\alpha} \bra{\partial_j \psi_\alpha} / \| \psi_\alpha \|_2^2$, 
\begin{align}
\|  P_{\rm osc}^{(j)} Q_\alpha^{(j)} g_\alpha \|_2 \leq C \| \psio \|_2^{-1} \| \partial_j \psio - \psi_\alpha \|_2 \| g_\alpha \|_2 \leq C \alpha^{-1/40}
\end{align}
from Lemma \ref{lemma:approx-E} {\it (d)} for sufficiently large $\alpha \geq \alpha_0$ and $\| \nabla \psio \| \geq c \alpha^{1/4}$ for some $c>0$. Thus we arrive at 
\begin{align}
\| P'_{\rm osc} g_\alpha \|_2 \leq C \alpha^{-1/40}. \label{eq:PoscPalpha}
\end{align}
On the other hand we have 
\begin{align}
\label{eq:hessian1}
{{\rm H}} = {\rm h}_{\sqrt{\alpha} \varphi_{\rm osc}} - \mu_{\psi_{\rm osc}} + ( \mu_{\psi_\alpha} - \mu_{\psi_{\rm osc}})
\end{align}
so that with Lemma \ref{lemma:approx-E} {\it{(d)}} we have $ \| {\rm H}_{\alpha}  P'_{\rm osc} f \|_2 \leq C \alpha^{1/2}$ and we arrive at 
\begin{align}
\label{eq:galpha-Q}
\expval{\rH - \widetilde{X}_{\psi_\alpha} }{  g_\alpha} - \expval{ \rH - \widetilde{X}_{\psi_\alpha} }{ Q_{\rm osc}' g_\alpha } \geq - C \alpha^{-19/40} \; .
 \end{align}
As a next step, we shall replace $Q_{\rm osc} g_\alpha$ with the $L^2$-normalized function 
\begin{align}
\label{def:gosc}
g_{\rm osc} := \frac{Q'_{\rm osc}  Q_\alpha' f }{  \| Q'_{\rm osc} Q_\alpha' f\|_2} \; 
\end{align}
For this we first observe that 
\begin{align}
 Q'_\alpha g_\alpha   = \frac{ \| Q'_{\rm osc} Q_\alpha' f\|_2 }{ \|  Q_\alpha' f\|_2} g_{\rm osc}  
\end{align}
and thus, we need to control the normalization constants' ratio. For this we use \eqref{eq:PoscPalpha} and find that 
\begin{align}
\| Q_\alpha' f \|_2^2 =& \| Q_{\rm osc}' Q_\alpha' f\|_2^2 + \| P'_{\rm osc} Q_\alpha' f \|_2^2 \leq  \| Q_{\rm osc}' Q_\alpha' f\|_2^2  +  C \alpha^{-1/12} \| Q_\alpha' f \|_2  
\end{align}
for a constant $C>0$ independent in $\alpha$. In particular we obtain 
\begin{align}
\label{eq:Qo-Qalpha}
\| Q_\alpha' f \|_2 \geq C (1- \alpha^{-1/40} )^{-1}  \| Q_{\rm osc}' Q_\alpha' f\|_2   \; 
\end{align}
for sufficiently large $\alpha \geq \alpha_0$. Thus from \eqref{eq:galpha-Q} and \eqref{def:gosc} we get 
\begin{align}
\label{eq:Hessian-step1}
 \expval{\rH - \widetilde{X}_{\psi_\alpha} }{  g_\alpha}  \geq  C (1- \alpha^{-1/40} )^{-2} \expval{ \rH - \widetilde{X}_{\psi_\alpha} }{  g_{\rm osc} }  - C \alpha^{19/40} \; .
 \end{align}

Next we show that the operator $\widetilde{X}_{\psi_\alpha}$ contributes sub-leading (i.e. $o ( \sqrt{\alpha})$) only. For this we write 
\begin{align}
\widetilde{X}_{\psi_\alpha} - \widetilde{X}_{\psio} = \left( \psi_\alpha (x) - \psio (x) \right) \widetilde{h}( x-y) \psi_\alpha (y) + \psio (x) \widetilde{h} (x-y) \left( \psi_\alpha (y) - \psio(y) \right) \; . 
\end{align} 
With $\vert \widetilde{h} (x) \vert \leq C x^2$ we find from Lemma \ref{lemma:approx-E} {\it (c)} and $\| g_{\rm osc} \|_2 =1$ that 
 \begin{align}
 \expval{\widetilde{X}_{\psi_\alpha }}{g_{\rm osc}} \geq  \expval{\widetilde{X}_{\psio }}{g_{\rm osc}}  - C \alpha^{9/20} \; . 
\end{align}  
We recall that $g_{\rm osc}$ is orthogonal to $\psio$ and its partial derivatives. In particular, as $\nabla \psio = x \psio$, the function $g_{\rm osc}$ is orthogonal to $x \psio$, too. Therefore not only the zero-th but also the first-order term of the Taylor expansion of $h$ in $\widetilde{X}_{\psio}$ evaluated in $g_{\rm osc}$ vanishes, i.e.  
 \begin{align}
 \expval{\widetilde{X}_{\psio }}{g_{\rm osc}}  =  \expval{\widetilde{\widetilde{X}}_{\psio }}{g_{\rm osc}}
 \end{align}
 where $\widetilde{\widetilde{X}}_{\psio}(x,y) =  \alpha \psio (x) ( h (x-y) - \| g \|_2^2 -\| \nabla g \|_2^2  (x-y)^2 ) \psio (y)$. Since $\vert \widetilde{h}(x) -\| \nabla g \|_2^2  x^2\vert \leq C  x^4$ by Ass. \ref{ass:reg} we find with the harmonic oscillators scaling properties that $\expval{\widetilde{\widetilde{X}}_{\psio }}{g_{\rm osc}} \geq - C $, and thus from \eqref{eq:Hessian-step1}
 \begin{align}
 \label{eq:Hessian-step2}
  \expval{ \rH - \widetilde{X}_{\psi_\alpha }}{g_\alpha} \geq  C (1- \alpha^{-1/40} )^{-2}  \expval{\rH }{g_{\rm osc}}  - C \alpha^{19/40}
 \end{align}
 for sufficiently large $\alpha \geq \alpha_0$. Now it remains to compare the r.h.s. with the harmonic oscillator. For this we split the operator $\rH$ into one part that is localized on a ball $B_R$ with $R= \alpha^{-1/6}$ (that we shall show is bounded from below by $\mathcal{H}_{\rm osc}$ that is $O( \sqrt{\alpha})$) and a part outside $B_R^c$ (that we will show is bounded from below by a positive constant of $O( \alpha)$, i.e. trivially satisfying the claim for sufficiently large $\alpha \geq \alpha_0$). 
 
For the localization we consider a partition of unity $0 \leq \eta^1, \eta^2 \leq 1$ with $\eta^i \in C_0^\infty ( \mathbb{R}^3)$ and 
\begin{align}
\eta^1 (x) = \begin{cases}
1 & x \in B_1 \\
0 & x \in B_2^c \\
\end{cases}, \quad \eta^2 = \sqrt{1- \vert \eta^1 \vert^2} \; . 
\end{align}
and define the rescaled version $\eta_R^i (x) = \eta^i (x/R )$ and the $L^2$-normalized function
\begin{align}
\label{eq:def-part1}
g^i_R = \eta_R^i g_{\rm osc} / \| \eta_R^i g_{\rm osc} \|_2  \; . 
\end{align}
With standard arguments of IMS localization we find 
\begin{align}
\expval{{ \rm H}_\alpha }{g_{\rm osc}} =& \sum_{i=1}^2 \| \eta_R^i g_{\rm osc} \|_2^2 \expval{\rH }{g_R^i} - \sum_{i=1}^2 \expval{\vert \nabla \eta_R^i\vert^2}{ g_{\rm osc}}\; . 
\end{align}
By scaling the last summand is of order $R^{-2} = \alpha^{1/3}$ yielding 
\begin{align}
\expval{{ \rm H}_\alpha }{g_{\rm osc}} \geq& \sum_{i=1}^2 \| \eta_R^i g \|_2^2 \expval{\rH}{g_R^i} - C\alpha^{1/3} \label{eq:hessian3} 
\end{align}
and it remains to estimate the expectation value $\expval{\rH}{f_R^i} $ for $i=1,2$. 

We start with the expectation value w.r.t $g_{R}^2$ supported on $B_R^c$. Since $\cos(k \cdot x ) \leq 1 $, we find 
\begin{align}
 \expval{{ \rm H}_\alpha }{g_R^2}  \geq& - \mu_{\psi_\alpha}-  \expval{\sqrt{\alpha} V_{\sqrt{\alpha}\varphi_\alpha}}{g_R^2} \notag \\
\geq& 2 \alpha \| g \|_2^2  + \expval{\sqrt{\alpha} V_{\sqrt{\alpha}\varphi_\alpha}}{g_R^2}  \; .\label{eq:hessian2}
\end{align} 
To show that the remaining term contributes sub-leading (i.e. $o(\alpha)$) only, we need to control the $L^\infty$- norm of $V_{\sqrt{\alpha}\varphi_\alpha}$ on $B_R^{c}$, i.e. 
\begin{align}
\sqrt{\alpha}\vert  V_{\sqrt{\alpha}\varphi_\alpha} (x) \vert   \leq  C \alpha \int  \vert h( x-y )\vert  \vert \psi_{\alpha}  (y)\vert^2 dy + C \alpha^{5/12} \; 
\end{align}
for $\vert x \vert \geq \alpha^{-1/6}$. We split the integral in $B_{\widetilde{R}}$ and $B_{\widetilde{R}}^c$ where now we choose $\widetilde{R} = \alpha^{-1/5}$. For $y \in B_{\widetilde{R}}^c$ we find by the scaling properties of $\psi_\alpha$ that 
\begin{align}
\| \psi_{\alpha} \|_{L^2( B_{\widetilde{R}}^c)}^2 \leq C \alpha^{2/5} \| x \psi_{\alpha} \|_{L^2( B_{\widetilde{R}}^c)}^2 \leq \alpha^{-1/10}
\end{align}
and we arrive with $\| h \|_{L^\infty ( \mathbb{R}^3)}\leq C$ for $\vert x \vert \geq \alpha^{-1/6} $ at 
\begin{align}
\sqrt{\alpha}\vert  V_{\sqrt{\alpha}\varphi_\alpha} (x) \vert   \leq  C \alpha \int_{\vert y \vert \leq \alpha^{-1/5}}   \vert h( x-y )\vert  \vert \psi_{\alpha}  (y)\vert^2 dy + C \alpha^{9/10} \; . 
\end{align}
Now let $ \vert y\vert \leq \alpha^{-1/5}$ and $\vert x \vert \geq \alpha^{-1/6}$. Then we have $\vert x - y \vert \geq \vert x \vert - \vert y \vert \geq C \alpha^{-1/6}$ for sufficiently large $\alpha \geq \alpha_0$. Thus,  with $ \| h \|_{L^\infty ( \mathbb{R}^3)} \leq C $ 
\begin{align}
\sqrt{\alpha}\vert  V_{\sqrt{\alpha}\varphi_\alpha} (x) \vert   \leq  C \alpha^{1+1/6}  \int_{\vert y \vert \leq \alpha^{-1/5}}  \vert x- y \vert   \vert \psi_{\alpha}  (y)\vert^2 dy + C \alpha^{5/12}
\end{align}
for $\vert x \vert \geq \alpha^{-1/6}$. With Cauchy Schwarz inequality and $\| \psi_{\alpha} \|_4 \leq C \| \psi_{\alpha} \|_{H^1} \leq C \alpha^{1/4}$ (from Lemma \ref{lemma:approx-E} {\it (d)}) we find 
\begin{align}
\label{eq:hessian2-inside}
\sqrt{\alpha}\vert  V_{\sqrt{\alpha}\varphi_\alpha} (x) \vert   \leq  C \alpha^{1+1/6 + 1/4-1/2}  + C \alpha^{5/12} \leq C \alpha^{11/12}  \; . 
\end{align}
Hence we deduce from \eqref{eq:hessian2}  
\begin{align}
\expval{{ \rm H}_\alpha}{g_R^2}  \geq 2 \alpha \| g \|_2^2 - C_2 \alpha^{11/12} \geq C_1 \alpha \label{eq:hessian-f2}
\end{align}
for constant $C_1,C_2$ independent of $\alpha$ and sufficiently large $\alpha \geq \alpha_0$. 

We recall that the goal for the expectation value 
\begin{align}
\expval{{ \rm H}_\alpha}{g_R^1} 
\end{align}
with $f_R^1$ supported on $B_R$ is a comparison with the Hessian of the harmonic oscillator $\mathcal{H}_{\rm osc}$ that is $O( \sqrt{\alpha})$. 
For this we observe that $f_R^1$ is almost orthogonal to $\psio$ and its partial derivatives as 
\begin{align}
\| P_{\rm osc} g_R^1 \|_2 \leq C \| \psi_{\rm osc} \|_{L^2( B_R^c)} \leq C \alpha^{-1/4} 
\end{align}
by Lemma \ref{lemma:approx-E} {\it (c)}) and similarly for the partial derivatives. Thus (with similar arguments as in the beginning of this proof (see Eq. \eqref{eq:Qosc} and subsequent)) instead of $g_R^1$ we consider in the following  the $L^2$-normalized function 
\begin{align}
\widetilde{g}_R^1 := \frac{Q_{\rm osc}' g_R^1}{\| Q_{\rm osc}' g_R^1 \|_2 } \;  
\end{align}
paying a price sub-leading in $\alpha$ (i.e. $o( \sqrt{\alpha})$ and given by 
\begin{align}
\label{eq:hessian4}
\bra{g_R^1} {{\rm H}}_\alpha  \ket{ f_R^1} \geq (1- \alpha^{-1/12})^2 \bra{\widetilde{g}_R^1}{{ \rm H}}_\alpha  \ket{ \widetilde{g}_R^1}  - C \alpha^{1/3} \; . 
\end{align}
We use the Taylor expansion of $\ch$, $\expval{x}{\psi_\alpha}=0$ and Lemma \ref{lemma:approx-E} {\it (d)} and find 
\begin{align}
\expval{{ \rm H}_\alpha }{\widetilde{g}_R^1} = \expval{ {\rm h}_{\sqrt{\alpha}\varphi_\alpha} - \mu_{\psi_\alpha}}{\widetilde{g}_R^1} \geq \expval{ {\rm h}_{\rm osc} - \bra{\psi_\alpha} {\rm h}_{\rm osc}\ket{\psi_\alpha}  }{\widetilde{f}_R^1} - C \alpha^{1/3} \; 
\end{align}
and thus (since $\bra{\psi_\alpha} {\rm h}_{\rm osc}\ket{\psi_\alpha} \geq \bra{\psio} {\rm h}_{\rm osc}\ket{\psio}$ )
\begin{align}
 \expval{{ \rm H}_\alpha }{\widetilde{g}_R^1} \geq  \expval{{ \rm H}_{\rm osc} }{\widetilde{g}_R^1} - C \alpha^{1/3} \; . 
\end{align}
By construction $\widetilde{g}_R^1$ is a $L^2$-normalized function and orthogonal to the harmonic oscillator's ground state and its partial derivatives. Thus $\widetilde{g}_R^1$ is a competitor for a minimizer of the harmonic oscillator's Hessian and we conclude that 
\begin{align}
\label{eq:hessian-f1}
 \expval{{ \rm H}_\alpha }{g_R^1} \geq (1- \alpha^{-1/40})^2 \mathcal{H}_{\rm osc} \; - C \alpha^{19/40} \; . 
\end{align}
Since $\mathcal{H}^{(2)}$ is a convex combination of \eqref{eq:hessian-f2} and \eqref{eq:hessian-f1}, we find 
\begin{align}
\mathcal{H}^{(2)} \geq C  (1- \alpha^{-1/40})^4 \mathcal{H}_{\rm osc}  - C \alpha^{19/40}
\end{align}
and conclude that there exists $\alpha \geq \alpha_0$ such that $\mathcal{H}^{(2)} \geq  C \alpha^{1/2}$ for all $\alpha \geq \alpha_0$. 
\end{proof}

The Hessian's positivity in the strong coupling limit $\alpha \rightarrow \infty$ leads to local coercivity estimates summarized in the following Corollary. 

\begin{cor}
\label{cor:loc-coerc}
There exists $\alpha_0 >0$ and $\kappa, C >0$ (independent of $\alpha$) such that for all $\alpha > \alpha_0$ and $\psi_\alpha \in \mathcal{M}_{\cE_\alpha}$, any $L^2$-normalized $\psi \in  H^1 ( \mathbb{R}^3)$ and $\varphi \in L^2 ( \mathbb{R}^3)$ with 
\begin{align}
\dist_{L^2} \left( \Theta ( \psi_\alpha ), \; \psi \right) \leq \kappa \alpha^{-1/2} 
\end{align}
we have 
\begin{align}
\cG_\alpha ( \psi, \varphi) - e_\alpha \geq C\sqrt{\alpha} \dist_{L^2} \left( \Theta (\psi_\alpha) , \; \psi \right)^2 \; ,\label{eq:coerc-psi-loc}\\
\cG_\alpha (\psi,  \varphi ) - e_\alpha \geq \frac{C}{\sqrt{\alpha}} \dist_{L^2_{\sqrt{\varepsilon}}} \left( \Omega ( \varphi_\alpha ), \; \varphi \right)^2  \; . \label{eq:coerc-phi-loc} 
\end{align}
\end{cor}

The proof is based on an expansion of $\cG_\alpha$ around the ground state energy $e_\alpha$. In the following, we provide an expansion of $\cG( \psi, \varphi)$ which will be useful for later proofs.  For this,let $\delta_1 = \psi - \psi_\alpha$ and $\delta_2 = \varphi - \varphi_\alpha$
\begin{align}
\cG_\alpha( \psi, \varphi ) -  e_\alpha  
=&\cG_\alpha ( \psi, \varphi) -  \expval{{\rm h}_{\sqrt{\alpha}\varphi_\alpha}}{\psi_\alpha} -   \| \varepsilon^{1/2 } \varphi_\alpha \|_2^2  \notag\\
=& 2 \Re \bra{\delta_1} \h_{\sqrt{\alpha}\varphi_\alpha} \ket{\psi_\alpha} +\sqrt{\alpha} \expval{ V_{ \delta_2}}{\psi_\alpha} \notag\\
& + 2 \sqrt{\alpha} \Re \bra{\varepsilon^{1/2}\varphi_\alpha}\ket{\varepsilon^{1/2}\delta_2} + \expval{\h_{\sqrt{\alpha}\varphi_\alpha}}{\delta_1} \notag\\
&+ 2 \sqrt{\alpha} \Re \bra{\psi_\alpha}  V_{\delta_2} \ket{\delta_1} +  \| \varepsilon^{1/2}\delta_2 \|_2^2 + O( \sqrt{\alpha} \| \delta_1 \|_{H^1}^2 \| \delta_2 \|_2 ) \; . 
\end{align}
We observe that the sum of the second and third term vanish by the definition of the potential \eqref{def:pot} and $\varphi_\alpha = - \sqrt{\alpha}\sigma_{\psi_\alpha}$. For the last two terms of the r.h.s., we complete the square 
\begin{align}
2 \sqrt{\alpha} &\Re \bra{\psi_\alpha }V_{\delta_2} \ket{ \delta_1} + \| \varepsilon^{1/2}\Re \delta_2 \|_2^2 \notag\\
=&  \|\varepsilon^{1/2} \Re \delta_2 +2  (2\pi)^{3/2}  \sqrt{\alpha} \; v \varepsilon^{-1/2} \; \widehat{  ( \Re\;\delta_1) \psi_\alpha} \|_2^2 \notag\\
& - 4\alpha \expval{ X_{\psi_\alpha}}{\Re\delta_1}   \label{eq:quad-ergaenzen}
\end{align}
where $X_{\psi_\alpha}$ is defined in \eqref{def:X} so that we arrive at 
\begin{align}
\cG_\alpha ( \psi, \varphi) -  e_\alpha   
=&  \expval{\h_{\sqrt{\alpha}\varphi_\alpha} }{\Im\delta_1} + \| \varepsilon^{1/2} \Im \delta_2 \|_2^2 \notag\\
&+  2 \Re \bra{\delta_1} \h_{\sqrt{\alpha}\varphi_\alpha} \ket{\psi_\alpha}+   \expval{\left( \h_{\sqrt{\alpha}\varphi_\alpha} - X_{\psi_\alpha} \right)}{ \Re\delta_1} \notag\\
&+\|\varepsilon^{1/2} \Re \delta_2 +2  (2\pi)^{3/2}  \sqrt{\alpha} \; v \varepsilon^{-1/2} \; \widehat{  ( \Re\;\delta_1) \psi_{0}} \|_2^2 \notag\\
&+ O(\sqrt{ \alpha}  \| \delta_1 \|_2^2 \| \delta_2 \|_2 )  \; . \label{eq:exp-G-tw}
\end{align}
The Euler-Lagrange equation of $\psi_\alpha$ together with the notation \eqref{def:X}, \eqref{def:X} and the observation that by $L^2$-normalization of $\psi_\alpha$ and $\psi$ 
\begin{align}
1 = \| \psi \|_2^2 = \| \psi_\alpha + \delta_1 \|_2^2 = 1 + \| \delta_1 \|_2^2 + 2 \Re \bra{\delta_1}\ket{\psi_\alpha} 
\end{align}
and therefore 
\begin{align}
\label{eq:normalization}
2 \Re \bra{\delta_1}\ket{\psi} = - \| \delta_1 \|_2^2  \; 
\end{align}
we find with \eqref{def:X} 
\begin{align}
\cG_\alpha ( \psi, \varphi) -  e_\alpha  
=&  \expval{{\rm H}_{\alpha} }{Q_\alpha\Im\delta_1} + \| \varepsilon^{1/2} \Im \delta_2 \|_2^2 \notag\\
&+ \expval{Q_\alpha\left( {\rm H}_{\alpha} - 4X_{\psi_\alpha} \right)Q_\alpha}{ \Re\delta_1} \notag\\
&+\|\varepsilon^{1/2} \Re \delta_2 +2  (2\pi)^{3/2}  \sqrt{\alpha} \; v \varepsilon^{-1/2} \; \widehat{  ( \Re\;\delta_1) \psi_\alpha } \|_2^2 \notag\\
&+ O(\sqrt{ \alpha}  \| \delta_1 \|_2^2 \| \delta_2 \|_2 ) +  O (\sqrt{\alpha} \| \delta_1 \|_2^3) + O ( \alpha \| \delta_1  \|_2^4) \; .  \label{eq:exp-G}
\end{align}

\begin{proof}[Proof of Cor. \ref{cor:loc-coerc}] In order to prove \eqref{eq:coerc-psi-loc} first, we remark that it suffices to consider $\psi \in H^1 ( \mathbb{R}^3)$  such that  
\begin{align}
\Im \bra{ e^{i \theta}\psi_\alpha^{y}} \ket{\psi} = 0, \quad \text{and} \quad \Re \bra {e^{i\theta}\psi_\alpha^{y}} \ket{\nabla \psi} = 0 \label{eq:cond}
\end{align}
hold. In particular we assume w.l.o.g. that $\theta =0$ and $y = 0$. 

Furthermore, completing the square we get 
\begin{align} 
\cG_\alpha ( \psi, \varphi) -e_\alpha   \geq& \cE_\alpha ( \psi ) - e_\alpha
\end{align}
and thus it suffices to consider the case $\delta_2 = \sqrt{\alpha} ( \sigma_\psi - \sigma_{\psi_\alpha})$. It follows from \eqref{eq:exp-G} and \eqref{eq:cond} 
\begin{align}
\cE_\alpha ( \psi)  -e_\alpha \geq  \expval{{\rm H}_{\alpha} }{\Im\delta_1} +  \expval{Q_\alpha ' \left(  {\rm H}_{\alpha} - 4 X_{\psi_\alpha} \right) Q'_\alpha }{ \Re\delta_1} + O ( \alpha \| \delta_1 \|_2^3) \; . 
\end{align}
 We recover back the Hessian of $\cE_\alpha$ which by Lemma \ref{lemma:Hessian} is positive for sufficiently large $\alpha \geq \alpha_0$. Moreover,  it follows from Lemma \ref{lemma:Hessian} that there exists $C_1>0$ (independent of $\alpha$) such that with $\| \delta_1 \|_2 \leq \delta \alpha^{-1/2}$ for sufficiently small $\delta>0$ by assumption, we have 
\begin{align}
\label{eq:E1}
\cE_\alpha ( \psi)- e_\alpha \geq C_1 \sqrt{\alpha} \| \delta_1 \|_2^2  \; . 
\end{align}
Moreover, since $\cos( k \cdot x) \leq 1 $ by Lemma \ref{lemma:approx-E} there exists a constants $\kappa_1,\kappa_2 >0$ such that 
\begin{align}\label{eq:E2}
{\rm h}_{\sqrt{\alpha}\varphi_\alpha} - e_\alpha \geq  -  \kappa_1 \Delta - \kappa_1 \sqrt{\alpha} \; . 
\end{align}
Interpolating between \eqref{eq:E1} and \eqref{eq:E2}, there exists a constant $C_2 >0$ such that for $\alpha \geq \alpha_0$ 
\begin{align}
\cE_\alpha ( \psi)- e_\alpha \geq C_2 \alpha^{1/4} \| \delta_1 \|_{H^1}^2  \; . 
\end{align}
By translational and rotational invariance of the energy, we conclude 
\begin{align}
\cE_\alpha ( \psi)- e_\alpha \geq C_2 \alpha^{1/4} \dist_{H^1} \left( \Theta (\psi_\alpha), \; \psi \right) \; . 
\end{align}
\newline 
Second we prove \eqref{eq:coerc-phi-loc}: Completing the square leads to 
\begin{align} 
\cG_\alpha ( \psi, \varphi) -e_\alpha  =& \cE_\alpha ( \psi ) - e_\alpha + \|  \Re \varphi + \sqrt{\alpha} \sigma_{\psi} \|_{L^2_{\sqrt{\varepsilon}}}^2 + \| \Im \varphi \|_{L_{\sqrt{\varepsilon}}^2}
\end{align}
so that we find from \eqref{eq:E1} that there exists $y \in \mathbb{R}^3$ and $\kappa_1 >0$ such that
\begin{align}
\cG_\alpha ( \psi, \varphi) -e_\alpha  \geq& \sqrt{\alpha} \kappa_1 \| \psi - \psi_\alpha^y \|_{2}^2 +  \| \Re \varphi + \sqrt{\alpha} \sigma_{\psi}   \|_{L^2_{\sqrt{\varepsilon}}}^2 +  \| \Im \varphi \|_{L^2_{\sqrt{\varepsilon}}}^2 \; . 
\end{align}
By regularity of $\varepsilon, v$, there exists $\kappa_3 >0$ such that 
\begin{align} 
\cG( \psi, \varphi) -e_\alpha  \geq& \kappa_3 \sqrt{\alpha} \; \|  \sigma_{\psi} - \sigma_{\psi_\alpha^y}  \|_{L^2_{\sqrt{\varepsilon}}}^2+  \| \varphi + \sqrt{\alpha} \sigma_{\psi}  \|_{L^2_{\sqrt{\varepsilon}}}^2
\end{align}
and we find by completing the square 
\begin{align} 
\cG( \psi, \varphi) -e _\alpha  \geq& \| ( 1 + \kappa_3/\sqrt{\alpha} )^{1/2} \left( \sigma_{\psi} - \sigma_{\psi_\alpha^y} \right)  - ( 1 +\kappa_3/ \sqrt{\alpha} )^{-1/2} \sqrt{\alpha} \left(\varphi +\sqrt{\alpha} \sigma_{\psi_\alpha^y} \right) \|_{L^2_{\sqrt{\varepsilon}}}^2  \notag\\
&+ \frac{ \kappa_3}{\sqrt{\alpha} +  \kappa_3} \| \varphi + \sqrt{\alpha } \sigma_{\psi_\alpha^y} \|_{L^2_{\sqrt{\varepsilon}}}^2  +  \|  \Im \varphi \|_{L^2_{\sqrt{\varepsilon}}}^2  \notag \\
\geq&   \frac{\kappa_4}{\sqrt{\alpha}} \;  \| \varphi + \sqrt{\alpha} \sigma_{\psi_\alpha^y}  \|_{L^2_{\sqrt{\varepsilon}}}^2   +  \| \Im \varphi \|_{L^2_{\sqrt{\varepsilon}}}^2  \label{eq:exp-G-phi}
\end{align}
for a constant $\kappa_4 >0$. We conclude that 
\begin{align} 
\cG( \psi, \varphi) -e _\alpha  \geq& \frac{\kappa_4}{\sqrt{\alpha}}  \dist_{L^2_{\sqrt{\varepsilon}} }   \left( \Omega ( \varphi_\alpha ), \; \Re \varphi \right)  \; . 
\end{align}
\end{proof}


\subsubsection*{Global Coercivity estimates} The Hessian's positivity shows the validity of global coercivity estimates of the energy. For this, we additionally have to assume that the ground state $\psi_\alpha$ is unique up to translations and phase, i.e. that $\mathcal{M}_{\mathcal{E}_\alpha} = \Theta ( \psi_\alpha)$. We remark that to prove the ground states uniqueness up to translations and phase by the local coercivity estimates Corollary \ref{cor:loc-coerc}, one needs an improved approximation of the ground state than in to Prop. \ref{prop:gs}.


\begin{cor}
\label{cor:coerc}
Assume that the ground state $\psi_\alpha$ of $\mathcal{E}_\alpha$ is unique up to translations and rotations. There exists universal constants $\alpha_0 \geq 0$ and  $C >0$ (independent of $\alpha$) such that for any $L^2$-normalized $\psi \in  H^1 ( \mathbb{R}^3)$ and $\varphi \in L^2 ( \mathbb{R}^3)$ with 
we have for all $\alpha \geq \alpha_0$
\begin{align}
\cG_\alpha ( \psi, \varphi) - e_\alpha \geq C\sqrt{\alpha} \dist_{L^2} \left( \Theta (\psi_\alpha) , \; \psi \right)^2 \; , \label{eq:coerc-psi}\\
\cG_\alpha (\psi,  \varphi ) - e_\alpha \geq \frac{C}{\sqrt{\alpha}} \dist_{L^2_{\sqrt{\varepsilon}}} \left( \Omega ( \varphi_\alpha ), \; \varphi \right)^2 \; .  \label{eq:coerc-phi}
\end{align}
\end{cor}

The proof follows the arguments presented in \cite[Lemma 2.6]{FRS_gap}. 

\begin{proof}
We first prove the global bound \eqref{eq:coerc-psi}. Then the second bound \eqref{eq:coerc-phi} follows similarly to the proof of Cor. \ref{cor:loc-coerc}.

In order to prove \eqref{eq:coerc-psi} we remark (similarly to the proof of Cor. \ref{cor:loc-coerc}) that is suffices to consider $\psi \in H^1 ( \mathbb{R}^3)$  such that  
\begin{align}
\Im \bra{ e^{i \theta}\psi_\alpha^{y}} \ket{\psi} = 0, \quad \text{and} \quad \Re \bra {e^{i\theta}\psi_\alpha^{y}} \ket{\nabla \psi} = 0 \label{eq:cond-2}\;. 
\end{align}
W.l.o.g. we assume $y=0 $ and $\theta=0$. By contradiction we assume that there does not exist a universal constant $C>0$ such that \eqref{eq:coerc-psi} holds. Then there exists a sequence of functions $\psi_n \in L^2( \mathbb{R}^3)$ with $\| \psi_n \|_{L^2( \mathbb{R}^3)} =1 $ such that
\begin{align}
\mathcal{E}_\alpha ( \psi_n ) \leq  e_\alpha + \frac{1}{n} \| \psi_n - \psi_\alpha \|_{H^1}^2 \leq \frac{2}{n} \| \psi_n \|_{H^1} - C \alpha \; . 
\end{align} 
It follows that $\mathcal{E}_\alpha ( \psi_n ) \geq \frac{1}{2} \| \nabla \psi_n \|_2^2 - C \alpha $. Therefore $\psi_n$ is uniformly bounded in $H^1$ and moreover a minimizing sequence. With similar arguments as in the proof of Lemma \ref{lemma:existence} $\psi_n$ converges to an element of the set of minimizes $\Theta (\psi_\alpha )$ given by \eqref{eq:cond-2} through $\psi_\alpha$. This is a contradiction since Cor. \ref{cor:loc-coerc} shows that locally coercivity estimates hold true. 
\end{proof}

Another consequence of the Hessian's approximate behavior is the following property.

\begin{cor} 
\label{cor:gap}
There exists $\alpha_0$ and a constant $C>0$ (independent of $\alpha$) such that for all $\alpha \geq \alpha_0$, we have $\rH - e_\alpha > C \sqrt{\alpha}$. 
\end{cor}

\begin{proof} The existence of a spectral gap of $\rH$ of order $\sqrt{\alpha}$ follows immediately from the global coervitiy estimates in Cor. \ref{cor:coerc}. 

\end{proof}

\subsection{Proof of Propositions \ref{prop:existence-and-unique},\ref{prop:gs}} 
\label{sec:proos-props}

In this section we prove Proposition \ref{prop:existence-and-unique}, \ref{prop:gs} based on the results proven before. 

\begin{proof}[Proof of Proposition \ref{prop:existence-and-unique}] The proposition follows immediately from Lemma \ref{lemma:existence}. 
\end{proof}

\begin{proof}[Proof of Proposition \ref{prop:gs}] The proposition follows from Lemma \ref{lemma:approx-E} and Corollary \ref{cor:coerc}.
\end{proof}

\section{Proof for traveling waves}
\label{sec:proofs-tw}

In this section, we prove Proposition \ref{prop:tw} on existence of subsonic traveling waves of the regularized Landau-Pekar equations. 

For this, we remark that it follows from the regularized polaron's dynamics that the traveling wave \eqref{def:tw} satisfies 
\begin{align}
\label{eq:tw}
- i \rv \cdot  \nabla \psi_\rv = \left( \h_{\varphi_\rv} + e_\rv \right) \psi_\rv, \quad  -\varepsilon^{-1}\rv \cdot k \varphi_\rv =  \varphi_\rv + \sqrt{\alpha} \sigma_{\psi_\rv} \; . 
\end{align} 

 \begin{proof} [Proof of Proposition \ref{prop:tw}] 
\noindent
\newline 
\textbf{Proof of {\it (a)}:} First, we prove the existence of traveling waves for sufficiently small velocities. Traveling wave solutions of \eqref{def:tw} are stationary points of the action functional $\cI_\rv $ given by
 \begin{align}
 \mathcal{J}_\rv (\psi, \varphi) := \expval{\h_{\varphi}}{\psi} + \| \varepsilon^{1/2} \varphi \|_2^2+e_\rv \| \psi \|_2^2 - \rv \cdot  \left(\expval{i\nabla} {\psi}+ \expval{p }{\varphi}\right) \; . 
 \end{align}
In the following we show that there exists a minimizer $(\psi_\rv, \varphi_\rv)$ of $\mathcal{J}_\rv$, and thus a traveling wave solution.  Since 
\begin{align}
\vert \rv \vert  \big\vert \expval{i \nabla}{\psi} \big\vert \leq \frac{1}{2} \| \nabla \psi \|_2^2  + 2 \rv^2 \| \psi \|_2^2, \quad 
\vert \rv \vert  \big\vert \expval{i k}{\varphi} \big\vert \leq \vert \rv \vert  \| \vert  p \vert^{1/2} \varphi \|_2^2   
 \end{align}
and for arbitrary $\delta>0$
 \begin{align}
\expval{V_\varphi}{\psi} \leq C  \| \varepsilon^{1/2} \varphi \|_2 \| \psi \|_2^2 \leq \delta \| \varepsilon^{1/2} \Re \varphi \|_2^2 + C_\delta \alpha \| \psi \|_2^4 \; , 
\end{align}
the action functional is bounded from below by 
 \begin{align}
 \mathcal{J}_\rv (\psi, \varphi ) &\geq \frac{1}{2} \| \nabla \psi \|_2^2 + \left( 1 - \delta \right) \| \varepsilon^{1/2} \varphi \|_2^2 - \vert \rv  \vert \| \vert p \vert^{1/2} \varphi \|_2^2 -C_\delta \alpha - C \rv^2 
 \end{align} 
We remark that in the last step we used the $L^2$-normalization of $\psi$.  By Ass. \ref{ass:super}, we have  $\varepsilon (k)\geq \cv \vert k \vert$ and thus, 
 \begin{align}
 \mathcal{J}_\rv (\psi, \varphi ) &\geq  \frac{1}{2} \| \nabla \psi \|_2^2 + \left(  \cv ( 1- \delta) - \rv \right) \| \vert p \vert^{1/2} \varphi \|_2^2 -C_\delta \alpha - C \rv^2 \; . 
 \end{align}
For $\vert \rv \vert \leq \cv - \delta$, it follows that any minimizing sequence $(\psi_n,\varphi_n)_{n \in \mathbb{N}}$ of $\cI_\rv$ is uniformly bounded in $H^1( \bR^3) \times L^2_{\sqrt{\vert \cdot \vert}}(\bR^3)$. Any minimizing sequence, thus, uniformly bounded and weakly converging in $H^1 ( \mathbb{R}^3) \times  L_{\sqrt{\varepsilon}} ( \mathbb{R}^3)$ to a limiting functional that is possibly zero (by translational invariance of the action). With similar arguments as after Eq. \eqref{eq:bound-E}, there exists a sequence $( \psi_n^{y_n}, e^{i p_n }\varphi_n)_{n \in \mathbb{N}}$ that converges strongly in $L^2(\bR^3) \times L^2_{\sqrt{\varepsilon}}( \bR^3)$ to a pair of non-zero limiting functions.  By semi-lower continuity of the $H^1$- and the $L_{\sqrt{\varepsilon}}$-norm, and 
\begin{align}
\vert \expval{V_{\varphi_n}}{\psi_n} -  \expval{V_{\varphi_\rv}}{\psi_\rv}\vert \leq C \| \psi_n - \psi_\rv \|_2 + C \| \varphi_n - \varphi_\rv \|_{L^2_{\sqrt{\varepsilon}}}
\end{align}
we conclude that the action functional $ \mathcal{J}_\rv$ attains its infimum for $( \psi_\rv, \varphi_\rv)$ which is a non-zero traveling wave solution \eqref{eq:tw}. 
\newline 
\textbf{Proof of scaling properties:} As a preliminary step towards proving Proposition \ref{prop:tw} {\em (b)} we shall first prove that
\begin{align}
\label{eq:tw-scaling}
\| x^2 \psi_\rv \|_2 \leq C \alpha^{-1/2} \; 
\end{align}
i.e. that the traveling wave satisfies similar scaling properties as the harmonic oscillator. We proceed similarly as in the proof of Lemma \ref{lemma:approx-E} {\it (b)}.  For this let ${\rm H}_0 := - \Delta /(2m) + i \rv \cdot \nabla$. Then the traveling wave equation \eqref{eq:tw} implies 
\begin{align}
\left( {\rm H}_0 + e_\rv \right)\psi_\rv = V_{\sqrt{\alpha}\varphi_\rv} \psi_\rv  \; . 
\end{align}
Since ${\rm H}_0 \geq - \rv^2/4$ and $e_\rv \geq - e_\alpha + \rv^2/4 $, the resolvent $\left( {\rm H}_0 + e_\rv \right)^{-1}$ is well defined and we can write 
\begin{align}
\label{eq:id-tw}
\psi_\rv = \left( {\rm H}_0 + e_\rv \right)^{-1}V_{\sqrt{\alpha}\varphi_\rv} \psi_\rv  \; . 
\end{align}
The resolvent's Green's function is given in terms of the inverse Fourier transform $\cF^{-1}$ of 
\begin{align}
\label{eq:Green}
G_{e_\rv} (z) = \cF^{-1} \left[  \big( p^2 /(2m)- \rv \cdot  p + e_\rv  \big)^{-1} \right] (z)
\end{align}
and can (by self-adjointness of ${\rm H}_0$ and functional calculus) explicitly computed. In fact by functional calculus we have 
\begin{align}
\left( H_0 + e_\rv \right)^{-1} = \int_0^\infty e^{- t ( {\rm H}_0 + e_\rv )  } \; dt = \int_0^t e^{- t e_\rv} e^{-t \left( p^2/(2m) - \rv \cdot p \right) } dt 
\end{align}
which leads with \eqref{eq:Green} to 
\begin{align}
G_{e_\rv} (z) &= \frac{1}{\left( 2 \pi \right)^{3/2}} \int_{\bR^3} \int_0^t e^{i z \cdot p} e^{-t e_\rv} e^{-t \left( p^2/(2m) - \rv \cdot p \right)} dt dp 
\end{align}
and we arrive with Fubini's theorem at 
\begin{align}
G_{e_\rv} (z) = (2m)^{3/2}e^{i z \cdot \rv }\int_0^t e^{- t \left(e_\rv- \frac{\rv^2}{4}\right)} \frac{e^{-\frac{mz^2}{2t}}}{ t^{3/2} } dt  = (2m)^{3/2} \frac{\sqrt{\pi}e^{i z \cdot \rv }}{\vert z \vert } e^{-\sqrt{4 m \left( e_\rv- \frac{\rv^2}{4}\right)} \vert z \vert}
\end{align}
for $\vert z \vert \not=0$.  We plug this identity into \eqref{eq:id-tw} and find that 
\begin{align}
\psi_\rv (x) = \int G_{e_\rv} (x-y) V_{\sqrt{\alpha} \varphi_\rv} (y) \; \psi_\rv (y) \; .
\end{align}
Thus the weighted $L^2$-norm,  we aim to find an upper bound for, becomes 
\begin{align}
\int x^{6} \vert  \psi_\rv (x) \vert^2 dx =& \int x^{6}  G_{e_\rv} (x-y) V_{\sqrt{\alpha} \varphi_\rv} (y) \; \psi_\rv (y) G_{e_\rv}( x- z) V_{\sqrt{\alpha} \varphi_\rv} (z) \psi_\rv (z) dxdydz 
\end{align}
that we can estimate by Cauchy Schwarz with 
\begin{align}
\int x^{6} \vert  \psi_\rv (x) \vert^2 dx 
\leq&  \int x^{2} \; \vert G_{e_\rv} (x-y)\vert^2  \vert V_{\sqrt{\alpha} \varphi_\rv} (y) \; \psi_\rv (y)  \vert^2 dxdy \notag\\
=& \pi\int \frac{ x^{6}}{ \vert x -y \vert^2 } e^{-2 \sqrt{4m \left( e_\rv- \frac{v^2}{4}\right)} \vert x-y \vert} \vert V_{\sqrt{\alpha}\varphi_\rv} (y) \psi_\rv (y) \vert^2 dydx \; . 
\end{align}
With 
\begin{align}
  \| V_{\sqrt{\alpha }\varphi_{\rv}} \psi_{\rv} \|_2^2 \leq C  \alpha \| \varphi_\rv \|_{L_{\sqrt{\varepsilon}}}^2 \| \psi_\rv \|_{2}^2 
\end{align}
and $\| \varphi_\rv \|_{L_{\sqrt{\varepsilon}}} \leq C \sqrt{\alpha}$ from Lemma \ref{lemma:dyn}, we can use a similar splitting of the above integral as in the proof of Lemma \ref{lemma:approx-E} {\it (b)} to then conclude by $e_\rv - \rv^2/4 \geq - e_\alpha \geq C \alpha $ with \eqref{eq:tw-scaling}. 

\textbf{Proof of {\rm (b)}:}  We observe that for $\rv = 0$ a traveling wave solution is given by $\psi_{\rv=0} = \psi_\alpha^y, \varphi_{\rv = 0} = \varphi_\alpha^y$ with $e_{\rv} = \mu_{\psi_\alpha^y}$ for any $y \in \mathbb{R}$. To prove \eqref{eq:tw-cont} it follows (similarly to the proof of Cor. \ref{cor:loc-coerc}) that it suffices to consider the decomposition 
\begin{align}
\label{eq:decomp-tw}
 \psi_{\rv} = \psi_\alpha^y + \delta_1, \quad \varphi_\rv = e^{i y \cdot }\varphi_\alpha + \delta_2, \quad \text{and} \quad e_{\rv} = \mu_{\psi_\alpha^y} + \mu_\rv \; 
\end{align}
with $\bra{\Re \delta_1}\ket{ \nabla \psi_\alpha^y} = 0$,  $\bra{\Im \delta_1}\ket{\psi_\alpha^y} =0$. In particular it follows from Cor. \ref{cor:coerc} and condition (i) that 
\begin{align}
\label{eq:decomp-tw-prop}
\| \delta_1 \|_2 \leq \kappa_1 \alpha^{-1/4}, \quad \| \delta_2 \|_{L_{\sqrt{\varepsilon}}} \leq \kappa_2 \alpha^{1/4} \quad \text{and} \quad \mu_\rv \leq \kappa_3 \sqrt{\alpha}
\end{align}
for sufficiently small $\kappa_1, \kappa_2, \kappa_3 >0$ (independent of $\alpha$).  In the following we assume w.l.o.g. that $y=0$. 

By definition of $\rH$ (see \eqref{def:X}) and the decomposition \eqref{eq:decomp-tw}, we can write the traveling wave equations \eqref{eq:tw} as
\begin{align}
-i \rv \cdot \nabla \left( \psi_\alpha + \delta_1 \right)   =& \left( V_{\sqrt{\alpha}\Re \delta_2} + \mu_{\rv} \right) \psi_\alpha + \left( { \rm H}_\alpha + V_{\sqrt{\alpha} \Re \delta_2} + \mu_\rv \right) \delta_1 \label{eq:tw1}\\
 \varepsilon^{-1}\rv \cdot  k \left( \varphi_\alpha + \delta_2 \right) =& \delta_2 + 2( 2 \pi )^{3/2} \sqrt{\alpha} v\varepsilon^{-1} \left( \widehat{\psi_\alpha \Re \delta_1} \right)  + \sqrt{\alpha} \sigma_{\delta_1} \;  \label{eq:tw2}
\end{align}
and it follows that the phase $\mu_\rv$ is given through the identity 
\begin{align}
\label{eq:tw-phase}
\mu_\rv ( 1 - \frac{1}{2}\| \delta_1 \|_2^2) =  \rv \cdot \bra{\psi_\alpha}\ket{\nabla \Im \delta_1} - \bra{\psi_\alpha} V_{\sqrt{\alpha} \Re \delta_2} \ket{\psi_\alpha} - \bra{\psi_\alpha} V_{\sqrt{\alpha} \Re \delta_2} \ket{ \Re \delta_1} \; .
\end{align} 
Plugging this identity back into \eqref{eq:tw1}, we get 
\begin{align}
\label{eq:delta-1}
\left(  \rH - A \right)  \delta_1 =& - Q_{\rv}  \left( V_{\sqrt{\alpha} \Re \delta_2} \psi_\rv  - \rv \cdot \nabla  \psi_\rv \right) 
\end{align}
where we introduced the notation $Q_\rv = 1- \ket{\psi_\rv}\bra{\psi_\rv} $ and the operator 
\begin{align}
A =&  \frac{\| \delta_1 \|_2^2}{2 - \| \delta_1 \|_2^2}   \left( \rv \cdot  \bra{\psi_\alpha}\ket{\nabla \Im \delta_1} - \bra{\psi_\alpha} V_{\sqrt{\alpha} \Re \delta_2} \ket{ \Re \delta_1}  -  \bra{\psi_\alpha} V_{\sqrt{\alpha} \Re \delta_2} \ket{\psi_\alpha} \right) \; . 
\end{align}
With the decomposition's properties \ref{eq:decomp-tw-prop} we find that 
\begin{align}
\| A \| \leq C \| \delta_1 \|_2^2 \left( \sqrt{\alpha} \| \delta_2 \|_2 +  \sqrt{\alpha}\rv \right) \leq C ( \alpha^{1/4} + \rv ) \; .
\end{align}
Thus by Cor. \ref{cor:gap} there exists $\kappa_4 >0$ such that by assumption $Q_{\rv}(\rH + A )Q_{\rv} \geq \kappa_4 \sqrt{\alpha} $, i.e. we can write  
\begin{align}
\delta_1 = \frac{Q_{\rv}}{\rH - A}  \left( V_{\sqrt{\alpha} \Re \delta_2} \psi_\rv  - \rv \cdot \nabla  \psi_\rv \right) \;. 
\end{align}
The second term of the r.h.s. leads with Proposition \ref{lemma:approx-E} to the desired bound. For the first term we observe that by definition of the potential and radialilty of $v$
 \begin{align}
 Q_{\rv} V_{\sqrt{\alpha} \delta_2} \psi_\rv =&   Q_{\rv} V_{\sqrt{\alpha} \delta_2^s} \psi_\rv
 \end{align}
 where $\delta_2^s$ denotes the symmetric part of $\delta_2$, i.e. $\delta_2^s(k) = \delta_2^s( -k )$. 
We observe that splitting $\delta_2$ into its symmetric $\delta_2^{s} (k) = \delta_2^{s} (-k)$ and anti-symmetric $\delta_2^{a} (k) = - \delta_2^{a} (-k)$ we have from the traveling wave equations \eqref{eq:tw2} 
\begin{align}
\delta_2^{s} = \frac{\varepsilon^{-2} ( k \cdot \rv)^2}{1-\varepsilon^{-2} ( k \cdot \rv)^2} \varphi_\alpha - \frac{\sqrt{\alpha}}{1-\varepsilon^{-2} ( k \cdot \rv)^2} \left(  2 (2 \pi)^{3/2} v\varepsilon^{-1} \widehat{\Re \delta_1 \psi_\alpha} + \; \sigma_{\delta_1} \right) \; . 
\end{align}
Here we used that $\varepsilon^{-2} ( k \cdot \rv)^2 \leq \rv^2/ \cv^2 <1$ by Ass. \ref{ass:super}. Hence 
\begin{align}
Q_{\rv} V_{\sqrt{\alpha} \Re \delta_2} \psi_\rv =& 2 \alpha Q_\rv \int  e^{ik \cdot }   \frac{\varepsilon^{-2} ( k \cdot \rv)^2}{1-\varepsilon^{-2} ( k \cdot \rv)^2} \varphi_\alpha (k) \;  dk \; \;   \psi_\rv \; \notag \\
&- \alpha Q_\rv \int e^{ik \cdot } \frac{ \widehat{h} (k) }{1-\varepsilon^{-2} ( k \cdot \rv)^2} \left( 2 \left( \widehat{\Re \delta_1 \psi_\alpha}\right) (k) + \; \widehat{\varrho}_{\delta_1} (k) \right) dk \; \; \psi_\rv \; . 
\end{align}
We observe that due to the projection $Q_\rv$ the first term of the Taylor expansion of $\cos ( k \cdot )$ vanishes. Thus with $\varepsilon^{-2} ( k \cdot \rv)^2 \leq \rv^2/ \cv^2 <1$ and $\bra{\Re \psi_\alpha} \ket{\psi_\alpha} = - \| \delta_1 \|_2^2$ we arrive at 
\begin{align}
\| V_{\sqrt{\alpha} \delta_2^s} \delta_1 \|_2 \leq& C\rv^2 \alpha   \left( \| x^2 \psi_\rv \|_2 + \| x \psi_\alpha \|_2 \| x \psi_\rv \|_2 \right) \notag \\
&+  C \alpha \left(  \| x^2 \psi_\rv \|_2 \|\delta_1 \|_2^2 + \| \delta_1 \|_2 \| x \psi_\alpha \|_2 \| x \psi_\rv \|_2 +  \| \delta_1 \|_2 \| x \delta_1 \|_2 \| x \psi_\rv \|_2\right)  \; . 
\end{align}
From the scaling properties of the traveling wave \eqref{eq:tw-scaling} and the ground state (see Cor. \ref{cor:gap}) we conclude 
\begin{align}
\| V_{\sqrt{\alpha} \delta_2^s} \delta_1 \|_2 \leq C \sqrt{\alpha} \left( \rv^2 + \kappa_1\alpha^{-1/4} \| \delta_1 \|_2 \right) 
\end{align}
yielding for $\vert \rv \vert \leq 1$
\begin{align}
\left( 1- \kappa_1 \alpha^{-1/4}  \right) \| \delta_1 \|_2 \leq C \vert \rv \vert 
\end{align}
and \eqref{eq:tw-cont-2} follows from \eqref{eq:tw2} resp. \eqref{eq:tw-phase}
\begin{align}
\| \delta_2 \| \leq  C  \sqrt{\alpha} \vert \rv \vert \quad \text{resp.} \quad \mu_{\rv} \leq \alpha^{3/4} \rv^2 \; . 
\end{align}

\end{proof}

\section{Proofs for definitions effective mass}
\label{sec:mass}
  In this section, we prove Theorem \ref{thm:mass-tw} and Theorem \ref{thm:mass} on the definition of the effective mass. 
  
 We remark that the proofs presented in this Section follow ideas from \cite{FRS_mass} where the non-regularized Landau-Pekar equations have been considered. For the non-regularized Landau-Pekar equations traveling waves are conjectured to not exists. However assuming their existence an energy expansion in the vein of the proof of Theorem \ref{thm:mass-tw} was sketched. Furthermore a different approach for a definition of the mass through an energy-velocity expansion was presented. The proof of Theorem \ref{thm:mass} given below uses ideas presented there. 
 
\subsection{Effective mass through traveling waves}
 
We consider the definition of the effective mass through subsonic traveling waves first (whose existence follow from Proposition \ref{prop:tw}). 

\begin{proof}[Proof of Theorem \ref{thm:mass-tw}] Let $\alpha \geq \alpha_0$ large enough and $\vert  \rv \vert  < \cv $. Then, by Proposition \ref{prop:tw} and Prop. \ref{prop:existence-and-unique}, there exists $y,z \in \mathbb{R}^3$ and $\theta \in (0, 2 \pi ]$ such that  
\begin{align}
\| e^{i \theta} \psi_\alpha^y - \psi_\rv \|_2 \leq C \rv, \quad \|e^{i z \cdot } \varphi_\alpha - \varphi_\rv \|_{L_{\sqrt{\varepsilon}}^2} \leq C \sqrt{\alpha} \vert \rv \vert \; 
\end{align}
with $\bra{\Im \psi_\rv}\ket{e^{i \theta} \psi_\alpha} = 0$, $\bra{\Re \nabla  \psi_\rv}\ket{e^{i \theta}\psi_\alpha^y} = 0$ and $\bra{e^{i z \cdot }\varphi_\alpha}\ket{\nabla \varphi_\rv} =0$ and $\psi_\alpha$ is uniquely given (up to translations and changes of phase). W.l.o.g. we assume in the following $y=0, \theta=0$. Then, \eqref{eq:exp-G-phi} shows (with similar arguments as used in \cite{FRS_mass}) that we it suffices to consider $z =0$, too. Thus, we decompose the traveling wave as 
\begin{align}
\left( \psi_\rv, \varphi_\rv \right) =  \left( \psi_\alpha +  \rv  \xi_\rv, \; \varphi_\alpha + \rv  \eta_\rv \right) 
\end{align}
with  $ \|\xi_\rv \|_2 \leq C$ and $\| \eta_\rv\|_{L^2_{\sqrt{\varepsilon}}} \leq C \sqrt{\alpha}$. Note that Proposition \ref{prop:tw} moreover shows that $e_\rv = \mu_{\psi_\alpha}  + O( \alpha^{3/4}\rv^2) $ and the linearisation of  the traveling wave equations \eqref{eq:tw} read 
\begin{align}
\begin{pmatrix}
\nabla \psi_\alpha \\  k \varphi_\alpha
\end{pmatrix}
 = 
 \begin{pmatrix}
 \rH & (2\pi)^{3/2} \sqrt{\alpha} \psi_\alpha \int dk \;  v(k) e^{ik \cdot }   \\ 
  (2\pi)^{3/2}   \sqrt{\alpha} \int dk \; v(k) e^{ik\cdot } \psi_\alpha  & \varepsilon 
 \end{pmatrix}
 \begin{pmatrix}
\xi_\rv \\ \eta_\rv  
\end{pmatrix} \; . 
\end{align}
In particular, it follows 
\begin{align}
& {\rm H}_{\alpha}  \Im \xi_\rv = \nabla \psi_\alpha \label{eq:tw-Imxi} \\
&\varepsilon \Im \eta_\rv = k \varphi_\alpha \label{eq:tw-Imeta}\\
&{\rm H}_{\varphi_\alpha}  \Re \xi_\rv +  \sqrt{\alpha} \psi_\alpha  V_{\Re \eta_\rv} = 0 \label{eq:tw-Rexi} \\
&  \sqrt{\alpha} V_{\psi_\alpha \Re \xi_\rv}  + \varepsilon \Re \eta_\rv = 0 \; .  \label{eq:tw-Reeta}
\end{align}
Combining \eqref{eq:tw-Reeta} and \eqref{eq:tw-Rexi}, we find 
\begin{align}
\left( {\rm H}_{\alpha} - 4 X_{\psi_\alpha} \right) \Re \xi_v = 0 \;. 
\end{align}
As ${\rm H}_{\alpha}$ is invertible on the span of $\partial_1 \psi_\alpha$,  we find from \eqref{eq:exp-G-tw} with \eqref{eq:tw-Imxi} and \eqref{eq:tw-Imeta} for all $ \vert \rv \vert \leq C \alpha^{-1} $ 
\begin{align}
\cG( \psi_\rv, \varphi_{\rv} ) = e_\alpha  + \left(  m +\frac{ 2 (2\pi)^3 \alpha}{3} \| k v \varepsilon^{-3/2} \; \widehat{\varrho}_{\psi_\alpha} \|_2^2  \right) \frac{\rv^2}{2} + O(\alpha \vert \rv\vert^3)  \;. 
\end{align}
\end{proof}

\subsection{Effective mass through energy-momentum expansion}

Here, we consider the definition of the effective mass through the energy-momentum expansion explained in Section \ref{sec:mass}.

\begin{proof}[Proof of Theorem \ref{thm:mass} ]  We first pick a trial state to show an upper bound on $\inf_{\cI_\rp}\cG (\psi, \varphi)$ which we use later for the lower bound. For this, we choose $\alpha_0 >0$ sufficiently large, such that by Proposition \ref{prop:existence-and-unique} {\rm (b)}, there exists a unique (up to translations and phases) pair of minimizers $(\psi_\alpha, \varphi_\alpha )$ of $\cG_\alpha$. 
\newline
\textbf{Proof of the upper bound of {\it (a)}:} It is easy to check that the trial states
\begin{align}
\psi_0  =&  ( 1- \mu_\rp ) \lambda_p^{-1} \psi_{\alpha} +i  \lambda_p \cdot  \; {\rm H}_{\alpha}^{-1} \nabla \psi_\alpha, \quad 
 \varphi_0  (k)   =   \varphi_\alpha (k)  + i (1- \mu_\rp) \lambda_p \cdot   k \;  (\varepsilon (k))^{-1} \varphi_\alpha  (k) 
\end{align} 
with the choice 
\begin{align}
\label{eq:lambda}
( 1- \mu_\rp)  \lambda_\rp  = \frac{\rp}{  m + \frac{ 2 (2\pi)^3 \alpha}{3} \| k v \varepsilon^{-3/2} \; \widehat{\varrho}_{\psi_\alpha}  \|_2^2} 
\end{align}
and $(1-\mu_\rp )^2 + \lambda_\rp^2 =1 $ (i.e. $\mu_p = O(  \alpha^{-2}\rp^2)$)  satisfy constraint \eqref{eq:mean} (using that ${\rm H}_{\alpha}^{-1} \nabla \psi_\alpha = m x \psi_\alpha$) and that for large $\alpha \geq \alpha_0$ 
\begin{align}
\| \psi_\alpha - \psi_0 \|_2 \leq C \alpha^{-1} \rp  \;  \quad \text{and} \quad \| \varphi_\alpha - \varphi_0 \|_2 \leq C \alpha^{-1/2} \rp  \; . 
\end{align}
With these observations, we plug the trail states into the expansion of $\cG_\alpha$ in \eqref{eq:exp-G} and find 
\begin{align}
\cG_\alpha ( \psi_0, \varphi_0) - e_\alpha  \leq&  ( 1- \mu_p)^2 \lambda_\rp^2  \left(  \expval{{\rm H}_{\varphi_\alpha}^{-1}}{ \nabla \psi_\alpha} +\frac{1}{3} \| \varepsilon^{-1/2}  k \varphi_\alpha  \|_2^2\right)  + O(\alpha^{-2} \rp^2)  \notag \\
=&   \frac{ ( 1- \mu_p)^2\lambda_\rp^2 }{2} \left(  m +  \frac{ 2 (2\pi)^3 \alpha}{3} \| k v \varepsilon^{-3/2} \; \widehat{\varrho}_{\psi_\alpha}  \|_2^2\right)  + O(\alpha^{-2} \rp^2)  \notag\\
=&    \left(  m +\frac{ 2 (2\pi)^3 \alpha}{3} \| k v \varepsilon^{-3/2} \; \widehat{\varrho}_{\psi_\alpha}  \|_2^2\right)^{-1} \frac{\rp^2}{2} + O( \alpha^{-2} \rp^2) \;  \label{eq:lb} . 
\end{align}
\noindent 
\newline 
\textbf{Proof of the lower bound of {\it (a)}:} For $ \rp \leq  \alpha^{1/4}$,  it follows from the upper bound and Cor. \ref{cor:coerc} that for any element $( \psi, \varphi) \in \cI_\rp$, there exists $y \in \mathbb{R}^3$ and $\theta, \omega \in (0, 2 \pi ]$ such that  
\begin{align}
\label{eq:lowerbound1}
\alpha^{-1/4} \| e^{i \theta} \psi_\alpha^y - \psi \|_2 = O( \alpha^{-1/2} \rp ) = \alpha^{1/4} \| e^{i \omega}\varphi_\alpha - \varphi \|_{L^2_{\sqrt{\varepsilon}}} \; . 
\end{align}
W.l.o.g. we assume $y=0, \theta =0$ and, with \eqref{eq:exp-G-phi} we consider furthermore $\omega =0$. It follows from the expansion \eqref{eq:exp-G} of the energy $\cG_\alpha$
\begin{align}
\cG_\alpha ( \psi, \varphi) - e_\alpha \geq& \expval{\Ho}{\Im \delta_1} + \| \varepsilon^{1/2} \Im \varphi \|_2^2 + O( \alpha^{-5/4} \rp^3 )   \; .  
\end{align}
Completing the square, we find with \eqref{eq:lambda} 
\begin{align}
\cG_\alpha ( \psi, \varphi) - e_\alpha \geq& \expval{{\rm H}_{\varphi_\alpha}}{\Im \delta_1 - \lambda_\rp{\rm H}_{\varphi_\alpha}^{-1} \nabla \psi_\alpha } \notag\\
&+ \expval{\varepsilon}{ \Im \varphi - \lambda_\rp \varepsilon^{-1}k \varphi_\alpha}  \notag\\
& + 2 \lambda_\rp  \left(  \braket{\nabla \psi_\alpha}{\Im \delta_1} +  \braket{ k \varphi_\alpha}{\varphi} \right)\notag \\
&- \lambda_\rp^2  \expval{\Ho^{-1}}{\nabla \psi_\alpha}- \lambda_\rp^2  \expval{\varepsilon^{-1} }{ k \varphi_\alpha} \notag\\
& + O( \alpha^{-5/4} \rp^3 ) \; . 
\end{align}
For the lower bound, we can neglect the first two lines and obtain 
\begin{align}
\cG( \psi, \varphi) - e_\alpha \geq& 2 \lambda_\rp ( 1- \mu_p)  \left(  \braket{\nabla \psi_\alpha}{\Im \delta_1} +  \braket{ k \varphi_\alpha}{ \Im \varphi} \right)\notag \\
&- \lambda_\rp^2( 1- \mu_p)^2  \expval{\Ho^{-1}}{\nabla \psi_\alpha}- \lambda_\rp^2  \expval{\varepsilon^{-1} }{ k \varphi_\alpha} \notag\\
&+ O( \alpha^{-5/4} \rp^3 )  \; . 
\end{align}
For the first line, we use the constraint (ii)$_{\rp}$, \eqref{eq:lowerbound1} together with Ass. \ref{ass:super} and the trivial bound $\| \nabla \psi_0 \| \leq C \sqrt{\alpha}$.The second line, we compute explicitly and obtain 
\begin{align}
\cG( \psi, \varphi) - e_\alpha \geq&   \left(  m +\frac{ 2 (2\pi)^3 \alpha}{3} \| k v \varepsilon^{-3/2} \; \widehat{\varrho}_{\psi_\alpha}  \|_2^2\right)^{-1} \frac{\rp^2}{2}  + O( \alpha^{-5/4} \rp^3 )  \; .\label{eq:ub}
\end{align}
Combining now the upper \eqref{eq:ub} and the lower bound \eqref{eq:lb}, we arrive at Theorem \ref{thm:mass}.

\textbf{Proof of {\rm (b)}:} The corresponding Lagrange functional to the minimization problem is given by 
\begin{align}
\label{eq:LM}
\mathcal{L}_\rp ( \psi, \varphi, \lambda, \mu ) := \mathcal{G}_\alpha( \psi, \varphi) - \lambda \left( \expval{i \nabla}{\psi} + \expval{ik}{ \varphi} - \rp \right) - \mu \left( \bra{\psi}\ket{\psi}  \right)  \; .
\end{align}
Thus any minimizer satisfies the traveling wave equations with velocity $\lambda$, i.e. 
\begin{align}
- i \lambda \partial_1 \psi_\rp = \left( h_{\varphi_\rp} - \mu \right) \psi_\rp , \quad \lambda k_1 \varphi_\rp = \varepsilon \varphi_\rp + ( 2 \pi ) ^{3/2} \sqrt{\alpha} v \widehat{\varrho}_{\psi_\rp} \; . 
\end{align}
By definition of the set $\mathcal{I}_\rp$ in \eqref{def:Ip}, Proposition \ref{prop:tw} shows that we can decompose 
\begin{align}
\psi_\rp = \psi_\alpha + \lambda \delta_1, \quad \varphi_\rp = \varphi_\alpha + \lambda \delta_2
\end{align}
with $\| \delta_1 \|_2 \leq C \alpha^{-1/4}$ and $ \| \varepsilon^{1/2} \delta_2 \|_2 \leq C \alpha^{1/4}$. The coercivity estimates from Corollary \ref{cor:coerc} together with the upper bound of part {\it (a)} then show  
\begin{align}
\lambda \| \delta_1 \|_2 \leq C \alpha^{-3/4}\rp, \quad \text{and} \quad \lambda \| \varepsilon^{1/2} \delta_2 \|_2 \leq C \alpha^{-1/4} \rp \;. 
\end{align}
Together with the traveling waves equation it follows from constraint (ii)$_\rp$ for all $\rp \leq 1$
\begin{align}
\rp = \lambda m_{\rm eff} + O( \alpha^{-1} \rp^2) = \lambda m_{\rm eff} + O (\alpha^{-1}\rp^2) \; . 
\end{align}
 Furthermore from the first part of the theorem and Theorem \ref{thm:mass-tw}, we conclude 
 \begin{align}
\mathcal{G}_\alpha( \psi_{\rv'},  \varphi_{\rv'} ) = E_\rp + O( \alpha^{-2} \rp^3) \; . 
 \end{align}
 where $( \psi_{\rv'},  \varphi_{\rv'} )$ denotes a traveling wave with velocity $\rv' =   m_{\rm eff}^{-1} \rp + O( \alpha^{-1} \rp^2)$. 
\end{proof} 

\subsection*{Acknowledgments} S.R. would like to thank Robert Seiringer for fruitful discussions, Krzysztof My\'sliwy for helpful remarks and the reviewer for careful reading and useful comments. Funding from the European Union's Horizon 2020 research and innovation program under the ERC grant agreement No. 694227 is gratefully acknowledged.

\subsection*{Data availability} Data sharing not applicable to this article as no datasets were generated or analyzed during the current study.

\section*{Declarations}

\subsection*{Conflict of interest} The authors declare that they have no conflict of interest.

\bibliographystyle{siam}

\begin{thebibliography}{10}


\bibitem{Varadhan} R.~Bazaes, C.~Mukherjee, S.R.S.~Varadhan, {\em Effective mass of the Fr\"ohlich Polaron and the Landau-Pekar-Spohn  conjecture}, Preprint: arXiv:2307.13058. 

\bibitem{Beetz} V.~Beetz, S.~Polzer, {\em Effective Mass of the Polaron: A Lower Bound}, Communications in Mathematical Physics,  399, p. 173-188 (2023). 

\bibitem{BS} M.~Brooks, R.~Seiringer, {\em The Fr\"ohlich Polaron at Strong Coupling -- Part II: Energy-Momentum Relation and Effective Mass},Preprint: arXiv:2211.03353.  

\bibitem{Gravejat} 
{ \sc F.~B\'ethuel, P.~Gravejat, J.-C.~Saut}, {\em Existence and properties of travelling waves for the {G}ross-{P}itaevskii equation}, Stationary and time dependent Gross-Pitaevskii equations, Contemp. Math., 473, Amer. Math. Soc., Providence, RI, (2008), pp. 55-103. 

\bibitem{DV} {\sc M.~Donsker and S.~Varadhan}, { \em Asymptotics for the polaron}, Comm. Pure Appl. Math., 36 (1983), pp. 505 -528.

 
 \bibitem{FS}
{\sc R.~Frank and B.~Schlein}, {\em Dynamics of a strongly coupled polaron},
  Letters in Mathematical Physics, 104 (2014), pp.~911--929. 

\bibitem{FS_gs} {\sc R. Frank and R. Seiringer}, {\em  Quantum corrections to the Pekar asymptotics of a strongly coupled polaron}, Commun. Pure Appl. Math., 74 (2021), pp. 544 - 588.

\bibitem{FG}
{\sc R.~Frank and G.~Zhou}, {\em Derivation of an effective evolution equation
  for a strongly coupled polaron}, Analysis and PDE, 10 (2017), pp.~379--422.


\bibitem{FRS_gap}
{\sc D.~Feliciangeli, S.~Rademacher, and R.~Seiringer}, {\em Persistence of the
  spectral gap for the {L}andau--{P}ekar equations}, Letters in Mathematical
  Physics, 111 (2021), pp.~1--19.


\bibitem{FRS_mass}
{\sc D.~Feliciangeli, S.~Rademacher, and R.~Seiringer}, {\em The effective mass problem for the {L}andau--{P}ekar equations},  Journal of Physics A: Mathematical and Theoretical,  55, 015201 (2022), 


\bibitem{FS} {\sc D.~Feliciangeli and R.~Seiringer } The Strongly Coupled Polaron on the Torus: Quantum Corrections to the {P}ekar Asymptotics. Archive for Rational Mechanics and Analysis, 242 (2021), pp. 1835--1906. 


\bibitem{G}
{\sc M.~Griesemer}, {\em On the dynamics of polarons in the strong-coupling
  limit}, Reviews in Mathematical Physics, 29 (2017).
  
\bibitem{Froe} 
{\sc H. Fr\"ohlich}, 
{\em Theory of electrical breakdown in ionic crystals},  Proc. R. Soc. Lond. A  \textbf{160}(901), 230--241 (1937).

\bibitem{Froehlich}
{\sc J.~Fr\"ohlich, B.L.G.~Jonsson and E.~Lenzmann}, { \em Boson Stars as Solitary Waves}. Commun. Math. Phys. 274, (2007), pp. 1-30. 


\bibitem{L}
{\sc E.~H. Lieb}, {\em Existence and uniqueness of the minimizing solution of
  {C}hoquard's nonlinear equation}, Studies in Applied Mathematics, 57 (1977),
  pp.~93--105.
  
\bibitem{Lieb} {\sc E.~H. Lieb}, {\em On the lowest eigenvalue of the Laplacian for the intersection of two domains }, Inventiones mathematicae ,74 (1983), pp.~ 441--48. 

\bibitem{LiebLoss}{ \sc E.~H. Lieb and M. Loss},  {\em Analysis }, American Mathematical Society, (2001). 
  
\bibitem{LMRSS}
{\sc N.~Leopold, D.~Mitrouskas, S.~Rademacher, B.~Schlein, and R.~Seiringer},
  {\em {L}andau-{P}ekar equations and quantum fluctuations for the dynamics of
  a strongly coupled polaron}, Pure Appl. Anal.  3 (4) 653 - 676, (2021). 
  
  
\bibitem{Landau}
{\sc L.~Landau}, {\em \"Uber die Bewegung der Elektronen im Kristallgitter}, 
Phys. Z. Sowjetunion, 3 (1933), 664.

\bibitem{LP}
{\sc L.~Landau and S.~Pekar}, {\em Effective mass of a polaron}, J. Exp. Theor.
  Phys, 18 (1948), pp.~419--423.
\bibitem{LRSS}
{\sc N.~Leopold, S.~Rademacher, B.~Schlein, and R.~Seiringer}, {\em The
  {L}andau-{P}ekar equations:{A}diabatic theorem and accuracy}, Analysis \& PDE 14 (2021),  2079-2100. 

\bibitem{LS}
{\sc E.~Lieb and R.~Seiringer}, {\em Divergence of the effective mass of a
  polaron in the strong coupling limit}, Journal of Statistical Physics, 180
  (2020), pp.~23--33.
  
\bibitem{LT} {\sc E.~Lieb and L.~Thomas}, {\em Exact ground state energy of the strong-coupling polaron}, Commun. Math. Phys, 183 (1997), p. 519.

\bibitem{M}
{\sc D.~Mitrouskas}, {\em A note on the {F}r\"ohlich dynamics in the strong
  coupling limit}, Letters in Mathematical Physics, 111 (2021).
  
\bibitem{MMS} {\sc D.~Mitrouskas, K.~My\'sliwy, and R.~Seiringer}, {\em Optimal parabolic upper bound for the energy-momentum relation of a strongly coupled polaron}, Preprint: arXiv:2203.02454 (2022). 

  
  \bibitem{MS} {\sc K.~My\'sliwy, and R.~Seiringer}, {\em Polaron Models with Regular Interactions at Strong Coupling}
Journal of Statistical Physics 186 (5) (2022)


\bibitem{Pekar}
{\sc S. Pekar}, Zh. Eksp. Teor. Fiz. 16 (1946), p.~335; 
J. Phys. USSR 10 (1946), p.~341.

\bibitem{Simon}
{\sc B.~Simon}, { \em Semiclassical analysis of low lying eigenvalues. I. Non-degenerate minima : asymptotic expansions}, Annales de l'I. H. P., section A, tome 38, no 3 (1983), p. 295-308

\bibitem{Spohn}
{\sc H.~Spohn}, {\em Effective mass of the polaron: {A} functional integral
  approach}, Annals of Physics, 175 (1987), pp.~278--318.
\end{thebibliography}

\end{document}